\RequirePackage{fix-cm}
\documentclass[final, smallcondensed]{svjour3}     
\smartqed  
\usepackage{graphicx}
\usepackage[authoryear]{natbib} 
\usepackage{comment}    
\usepackage{graphicx}  
\usepackage{amsmath}
\usepackage{amssymb}
\usepackage{mathrsfs}
\usepackage{mathtools}  
\usepackage{hhline}
\usepackage{comment}
\usepackage[noend]{algpseudocode}
\usepackage{algorithm}
\usepackage[hidelinks, colorlinks=true, linkcolor=black, citecolor=blue, urlcolor=black]{hyperref}

\usepackage{etoolbox}

\makeatletter

\patchcmd{\NAT@citex}
{\@citea\NAT@hyper@{%
		\NAT@nmfmt{\NAT@nm}%
		\hyper@natlinkbreak{\NAT@aysep\NAT@spacechar}{\@citeb\@extra@b@citeb}%
		\NAT@date}}
{\@citea\NAT@nmfmt{\NAT@nm}%
	\NAT@aysep\NAT@spacechar\NAT@hyper@{\NAT@date}}{}{}

\patchcmd{\NAT@citex}
{\@citea\NAT@hyper@{%
		\NAT@nmfmt{\NAT@nm}%
		\hyper@natlinkbreak{\NAT@spacechar\NAT@@open\if*#1*\else#1\NAT@spacechar\fi}%
		{\@citeb\@extra@b@citeb}%
		\NAT@date}}
{\@citea\NAT@nmfmt{\NAT@nm}%
	\NAT@spacechar\NAT@@open\if*#1*\else#1\NAT@spacechar\fi\NAT@hyper@{\NAT@date}}
{}{}

\makeatother

\newcommand{\raf}[1]{(\ref{#1})}

\newcommand{\OPT}{\ensuremath{\textsc{Opt}}}

\def\CC{\mathbb C}
\def\RR{\mathbb R}
\def\ZZ{\mathbb Z}

\def\cC{\mathcal C}

\def\cF{\mathcal F}
\def\cG{\mathcal G}
\def\cI{\mathcal I}
\def\cH{\mathcal H}
\def\cJ{\mathcal J}

\def\cL{\mathcal L}

\def\cN{\mathcal N}
\def\cS{\mathcal S}
\def\cT{\mathcal T}

\def\cU{\mathcal U}

\def\bzero{\mathbf 0}

\newcommand{\hide}[1]{}

\newcommand{\re}{\ensuremath{\mathrm{Re}}}
\newcommand{\im}{\ensuremath{\mathrm{Im}}}

\newcommand{\cP}{\ensuremath{\mathcal{P}}}
\newcommand{\cQ}{\ensuremath{\mathcal{Q}}}

\newcommand{\cV}{\ensuremath{\mathcal{V}}}
\newcommand{\cE}{\ensuremath{\mathcal{E}} }

\newcommand{\polylog}{\operatorname{polylog}}

\def\eps{\varepsilon}

\newcommand{\bone}{\ensuremath{\boldsymbol{1}}}

\usepackage{pifont}
\renewenvironment{proof}{{\noindent \bfseries \textit{Proof:}}}{}
\renewenvironment{remark}{{\noindent \bfseries \textit{Remark:}}}{}
\usepackage{fancyhdr}
\date{}

\begin{document}

\title{Approximations for Generalized Unsplittable Flow on Paths with Application to Power Systems Optimization\thanks{This work was supported by the Khalifa University under Awards No. CIRA-2020-286 and KKJRC-2019-Trans1.}}

\titlerunning{Approximations for Unsplittable Stairstep Flow on Paths}


\author{Areg Karapetyan \and Khaled Elbassioni\and Majid Khonji \and Sid Chi-Kin Chau}
\authorrunning{A.\, Karapetyan, K.\, Elbassioni , M.\, Khonji and S.\, C.-K. Chau}
%
\institute{A.\, Karapetyan \and K.\, Elbassioni \and M.\, Khonji \at Department of Electrical Engineering and Computer Science, Khalifa University, UAE
	\email{\{areg.karapetyan, khaled.elbassioni, majid.khonji\}@ku.ac.ae}
	\and
	S.\, C.-K. Chau \at
	Research School of Computer Science, Australian National University, Canberra, Australia
	\email{sid.chau@anu.edu.au}
}

\maketitle

\begin{abstract}
	The \textit{Unsplittable Flow on a Path} (UFP) problem has garnered considerable attention as a challenging combinatorial optimization problem with notable practical implications. Steered by its pivotal applications in \textit{power engineering}, the present work formulates a novel generalization of UFP, wherein demands and capacities in the input instance are monotone step functions over the set of edges. As an initial step towards tackling this generalization, we draw on and extend ideas from prior research to devise a quasi-polynomial time approximation scheme (QPTAS) under the premise that the demands and capacities lie in a quasi-polynomial range. Second, retaining the same assumption, an efficient logarithmic approximation is introduced for the single-source variant of the problem. Finally, we round up the contributions by designing a (kind of) black-box reduction that, under some mild conditions, allows to translate LP-based approximation algorithms for the studied problem into their counterparts for the \textit{Alternating Current Optimal Power Flow} (AC OPF) problem -- a fundamental workflow in operation and control of power systems.

%
%

	
\end{abstract}

{\section*{Nomenclature}
	
\vspace*{-20pt}	
\def\arraystretch{1.35}%
\begin{table}[H]
	\centering
	\begin{tabular}{|cl|}
		\hline
		\multicolumn{1}{|c|}{\textbf{Notation}} & \textbf{Description} \\ \hline
		\multicolumn{2}{r}{	\rule{0pt}{3ex}\textit{\fontsize{6.5}{6}\selectfont Summary of key notations related to the generalized UFP}}                            \\ \hline
		$\cG$& Line graph                   \\
		$\cV$& Set of vertices (indexed by $i$ or $j$)                   \\
		$\cE$& Set of $m$ edges (indexed by $e$ or $(i, j)$)                   \\
		$\cI$& Set of $n$ users (indexed by $k$)                    \\ 
		$\cQ$& Grouping (of cardinality Q) of users based on utility-to-demand ratio                      \\
		$\cI^{q}$& Set of users in group $q \in \cQ$                  \\
		$\cL^q$& Set of users with ``large'' demands in group $q \in \cQ$                      \\ 
		$\cS^q$& Set of users with ``small'' demands in group $q \in \cQ$                    \\ 
		\hline \hline
		
		$d$& Number of dimensions               \\
		$u_k$& User $k$'s utility value             \\
		$x_k$& Decision variable for user $k$          \\
		$f_k^r (\cdot)$& User $k$'s demand function over $\cE$ in dimension $r \in \{1,2, \hdots, d\}$                    \\
		$e_k^r, \hat{e}_k^r$& User $k$'s demand function's binding edges in dimension $r \in \{1,2, \hdots, d\}$         \\
		$c^r (\cdot)$& Capacity function over $\cE$ in dimension $r \in \{1,2, \hdots, d\}$                      \\ \hline\hline
		
		$T_1, \hdots , T_d$& $d$ positive integers              \\
		$T$& Maximum among $T_1, \hdots , T_d$              \\
		$C_1, \hdots , C_d$& $d$ integers each greater than $1$              \\
		$P_r$& Number of edge partitions in dimension $r \in \{1,2, \hdots, d\}$               \\
		$\epsilon$& Constant in $(0,1)$              \\\hline

		\multicolumn{2}{r}{	\rule{0pt}{3ex}\textit{\fontsize{6.5}{6}\selectfont Summary of key notations related to AC OPF}}                            \\ \hline
		$\cT$& Graph of a radial distribution network               \\
		$\cV^+$& Set of vertices excluding the root $0$             \\
		$\cV_i^+$& Set of vertices in $\cV^+$  excluding the node $i$  \\
		$\cN$& Set of all users (electrical loads) (indexed by $k$)                \\
		$\cU_j$& Set of users at node $j$                 \\
		$\cN_j$& Set of users residing on the subpath rooted at node $j$                 \\
		$\cF$& Set of users with elastic power demands               \\
		$\cP_j$& The (unique) path from node $j$ to the root $0$\\\hline \hline
		
		$s_k$& User $k$'s complex power demand             \\
		$z_{i,j}$& Impedance of power line $(i,j)$         \\
		$V_j$& Voltage at node $j$                   \\
		$v_j$& Voltage magnitude square at node $j$                   \\
		$I_{i,j}$& Current traversing through line $(i,j)$                  \\
		$l_{i,j}$& Squared magnitude of current flowing through line $(i,j)$          \\
		$S_{i,j}$& Complex power flowing from node $i$ to node $j$                    \\ \hline
	\end{tabular}
	\label{tab:my-table}
\end{table}}
\def\arraystretch{1}%

\section{Introduction}

The UFP, in its most generic form, takes as input a capacitated line graph along with a collection of flow requests, each parameterized by a demand, a profit (utility) and a pair of source-sink vertices. Constrained by edge capacities, the pursued objective is to compute a maximum profit subset of requests routable simultaneously. Despite its apparent simplicity, UFP specializes to a number of classical NP-hard combinatorial problems, including the Knapsack problem (when the graph comprises a solitary edge) and the Maximum Edge-disjoint Path problem (when all demands and capacities are set to unity). On the practical side, this problem underlies a spectrum of real-world applications in communication networks~\citep{Bar-Noy2001}, space missions~\citep{Hall1994}, the Web~\citep{Albers:1999} and data centers management~\citep{BCES06}, to name a few.

Recently, several studies have revisited UFP generalizing it from different perspectives. In~\citep{10.1007/978-3}, UFP is extended to the Storage Allocation problem, where the requests are additionally characterized by a vertical position (i.e., height) and a coupling constraint is imposed enforcing a non-overlapping drawing of them. Another line of work in~\citep{Adamaszek2018}, adapted the problem to the setting of a submodular objective function stimulated by theoretical and practical appeal thereof. Expanding the application scope further, this paper introduces a novel generalization of UFP. Previously, Cook et al. \citep{BFb0055088} developed an interesting framework linking electrical power transmission with the unsplittable flow on general graphs. Solidifying this nexus, we establish a formal bond between UFP and the AC OPF, an essential problem in power systems engineering introduced by Carpentier in 1962~\citep{C62} (see Section~\ref{elecflows} for particulars). Formally, the proposed generalization is defined in what follows.
\paragraph{\textbf{Generalization of UFP:}} In the {\it $d$-dimensional Unsplittable Stairstep Flow on a Path} ($d$-USFP) problem, defined here for a fixed positive integer $d\in\ZZ_+$, given is a line network $\cG=(\cV,\cE)$ rooted at node $0$ and a set $\cI$ of $n$ users. Assuming an ascending ordering of the edges by distance from the root {(i.e.,  $e_1<e_2<\ldots<e_m$, where $e_i=(i-1,i)$)}, each user demand is captured by a $d$-dimensional vector $f_k=(f_k^1,\ldots,f_k^d)$, where for $\forall ~r$, $f_k^r:\cE\to\RR_+$ are either monotone non-increasing or monotone non-decreasing step functions over $\cE$ (e.g., $f_k^r$ is monotone non-decreasing if $f_k^r(e)\le f_k^r(e')$ whenever $e\le e'$). As with UFP, if $f_k$ is satisfied (routed), $u_k \geq 0$ is the perceived utility for customer $k$ and each edge $e\in\cE$ is associated with a capacity, which in the current context is a $d$-dimensional vector $c=(c^1,\ldots,c ^d)$, where $c^r:\cE\to\RR_+$ is a monotone non-decreasing function on $\cE$. With this input, $d$-USFP then takes the form
\begin{align}
\big(d\text{-USFP}[\cI,c]\big)&\quad \max_{\substack{x}} \sum_{k\in\cI}u_kx_k,  \notag \\
\text{s.t.} \ & \quad \sum_{k\in \cI}f_k^r(e)x_k\le c^r(e),~\quad \forall~e\in\cE,~r\in\{1,...,d\}\label{ee0}\\
& \quad x_k \in \{0,1\},\quad \forall ~k \in \cI \,.\label{ee1}
\end{align}
In the above formulation, we may assume without loss of generality (by reversing the order on $\cE$ if necessary) that $f_k^r(\cdot)$ is monotone non-decreasing for $\forall~k \in \cI,r \in \{1,...,d\}$. While $d$-USFP can be defined for any such $f_k^r(\cdot)$, this paper confines the scope to functions of {\it separable form}. More precisely, given positive integers $T_1,\ldots,T_d$, monotone (non-decreasing) functions $b^{r,t}:\cE\to\RR_+$, for $t=1,\ldots,T_r, r\in[d]$ as well as non-negative numbers $a^{r,t}_k\in\RR_+$ and edges $e^{r}_k,\hat e^{r}_k\in\cE$, for $t=1,\ldots,T_r$, it is assumed that
\begin{align}
\label{funct}
f_k^r(e)=\sum_{t=1}^{T_r}a_k^{r,t}\tilde b^{r,t}_k(e),
\text{ where } 
\tilde b^{r,t}_k(e)=\left\{
\begin{array}{ll}
0&\text{ if }e< e^{r}_k,\\
b^{r,t}(e)&\text{ if }e^{r}_k\le e< \hat e^{r}_k\,.\\
b^{r,t}(\hat e^{r}_k)&\text{ otherwise.}
\end{array}
\right.
\end{align}
This choice of functions\footnote{Note that, for $\forall~r \in \{1,...,d\}$, the capacity function $c^r(\cdot)$ adheres to this form trivially with $T_r=a^{r,t}=1,  b^{r,t}(\cdot)=c^r(\cdot)$ for $\forall~t$ and $e^{r} = e_1, \hat e^{r}=e_m$.} stems from the relevant structural properties of OPF constraints, as elaborated in Section~\ref{sec:pre}. Yet, even with this condition in place, $d$-USFP remains substantially more complicated than UFP as it entails the packing of monotone step functions of special type (rather than intervals) within a given capacity function. Hence, known techniques for UFP, if amenable, have to be extrapolated in a non-trivial manner to deal with $d$-USFP. In the proceeding paragraphs, we briefly review these techniques. 

\paragraph{\textbf{Related Work:}}\label{subsec:litreview}

As noted previously, UFP is NP-hard since it specializes to the Knapsack problem. In fact, even under the setting of uniform profits and capacities, it has proven to be strongly NP-hard~\citep{Chrobak2012}. In light of this hurdle, most of the prior studies attempted simplified variants of UFP, with the two predominantly common ones being the uniform capacity UFP (UCUFP) and the UFP with the {\it no-bottleneck assumption} (UFP-NBA).

For UCUFP, which was also studied under the name of Resource Allocation problem, the first constant factor approximation was presented in~\citep{1099-1425}, attaining a $6$-approximation via LP rounding techniques. This factor was then refined by~\citep{3-540-47867} to ($2 + \epsilon$). The approach therein decouples the instance into small and large requests, subsequently tackling the former in a fashion analogous to~\citep{1099-1425}, while the latter through dynamic programming.

Ensuing from a more general case, UFP-NBA restricts the maximum demand to be at most the minimum capacity of any edge. The crux of this condition is rooted in the integrality gap of the natural LP relaxation of UFP, which was shown to be $\Omega(n)$ in~\citep{CCGK07}, whereas that of UFP-NBA is O$(1)$. For UFP-NBA, the first constant factor approximation was derived in~\citep{CCGK07}. Improving upon this, Chekuri et al.~\citep{CMS07} obtained a ($2 + \epsilon$)-approximation. These both studies broadly follow the aforementioned framework of decomposing the requests into small and large.

Turning to UFP, in 2006 Bansal et al.~\citep{BCES06} developed a deterministic QPTAS under the assumption that the capacities and demands are bounded by $2^{\polylog(n)}$, thereby ruling out UFP's APX-hardness and hinting to the likely existence of a PTAS. The first polynomial-time approximation algorithm for UFP, yielding O$(\log n)$ guarantee, was introduced in~\citep{BFKS14}. The algorithm is combinatorial, thus allowing to bypass the $\Omega(n)$ integrality gap of the natural LP relaxation. Later on, this result was extended in~\citep{anagnostopoulos2014mazing} to a $(2+\epsilon)$-approximation. {Beating the barrier of $2$, Grandoni et al.~\citep{Grandoni3188745} provided a dynamic programming-based polynomial time algorithm with an approximation factor of $(\frac{5}{3} + \epsilon)$, which in~\citep{1.9781611977073.39} was subsequently improved to $1 +\frac{1}{1+e} +\epsilon < 1.269$ (in expectation) via a novel randomized sketching technique. Closing the search for a PTAS, very recently Grandoni et al.~\citep{10.1145/3519935.3519959} devised a polynomial time $(1 + \epsilon)$-approximation algorithm which tackles UFP by rephrasing the problem as a solitary game and is the best possible result unless P=NP.}

\paragraph{\textbf{Contributions and Paper Outline:}} As such, this study advances extant research in the following two aspects:
\begin{itemize}
	\item[\ding{227}] We introduce a practically-driven generalization of UFP and initiate the search for its efficient approximations. As a first step in this direction, we extend the ideas in~\citep{BCES06} to construct a QPTAS for separable $d$-USFP, under the assumption that the demands and capacities lie in a quasi-polynomial range. Second, relying on the same assumption, we devise an LP-based O$(d\log n)$-approximation for the  single-source setting of the problem (i.e., when all the requests share the same origin). The algorithm hinges on a simple reduction allowing to transform the problem to an easier instance with only O$(d\log n)$ constraints.
	\item[\ding{227}] A (kind of) black-box reduction is derived that, under some practical assumptions, translates an LP-based approximation for separable $d$-USFP into its analog for OPF on line distribution networks. This result complements the strand of research in~\citep{7590153, 8637153, KHONJI201934, Elbassioni2019, 9301221} concerned with developing efficient approximations tailored for combinatorial optimization of AC electric power systems.
\end{itemize}

The remainder of this article is organized as follows. Section~\ref{sec:prelim} covers the adopted notation along with a basic result on partitioning of the studied step functions. Section~\ref{sec:qpt} presents the QPTAS for separable $d$-USFP. In Section~\ref{sec:logarithmic} we provide the logarithmic approximation for single-source separable $d$-USFP. Section~\ref{elecflows} contains an overview of AC OPF problem, followed by its mathematical formulation and the proposed reduction procedure producing LP-based approximations for OPF on line networks. Lastly, Section~\ref{conclrem} concludes the paper with a discussion on applications and connotations of present contributions as well as prospective directions for further developments.  

\section{Notational Convention and Preliminaries}\label{sec:prelim}

{In what follows, unless otherwise explicitly mentioned, constants or variables are denoted in normal font (e.g., $C$, $d$), while sets in calligraphic capital letters (e.g., $\cE$). We let $\bzero$ and $\bone$ symbolize the vectors of all zeros and ones, respectively, and as a shorthand, we shall write $[n]$ to encode the range $\{ 1,...,n \}$ for an integer $n$. Unless stated differently, we designate the operators $\bar{}$ , $\underline{~}$ to capture the maximum and minimum values of a variable/parameter/function, respectively. Given a complex number $\nu\in\CC$, we let $|\nu|$ be its {\it magnitude}, $\arg(\nu)$ be the {\it phase angle} that it makes with the real axis, $\nu^*$ be its complex {\it conjugate} and write $\nu^{\rm R} \triangleq {\rm Re}(\nu)$, $\nu^{\rm I} \triangleq {\rm Im}(\nu)$ for its real and imaginary components, respectively. With a slight abuse of notation, we shall also use the superscript $^*$ to mark the optimal solutions.}

In line with~\citep{BCES06}, we suppose the range of demands and capacities is quasi-polynomial. Mathematically,
$$
\displaystyle \max\left\{\frac{\max_{e\in\cE, k \in \cI, r\in [d]}f_k^r(e)}{\min_{e\in\cE, k \in \cI, r\in [d]: f_k^r(e) > 0}f_k^r(e)},\frac{\max_{e\in\cE, r\in [d]}c^r(e)}{\min_{e\in\cE, r\in [d]: c^r(e) > 0}c^r(e)}\right\}=2^{\polylog(n)}
. 
$$
This assumption is leveraged both, in the QPTAS and the logarithmic approximation, however, one can possibly discard it with techniques from~\citep{BGKMW15}.

The proposed approximations employ the following simple, yet crucial, lemma which, in a sense, states that the line can be partitioned into logarithmic (in $n$) number of regions such that, for each user $k$, the function $f_k(\cdot)$ is roughly constant in each region.    

\begin{lemma}\label{l5}For any $C_r>1$, $r\in[d]$,
	$\cE$ can be partitioned along each coordinate $r\in[d]$ into $P_r<T_r\log_{C_r} \big( \frac{\overline b^r}{\underline b^r}\big)$ intervals $\cE^r=\bigcup_{p=1}^{P_r}\cE_p^r$, where   $\cE_p^r:=\{e_{\underline{i}(p,r)},e_{\underline{i}(p,r)+1},\ldots,e_{\overline{i}(p,r)}\}$,
	and
	$$\cdots<e_{\overline{i}(p-1,r)}<e_{\underline{i}(p,r)}<e_{\underline{i}(p,r)+1}<\cdots<e_{\overline{i}(p,r)}<e_{\underline{i}(p+1,r)}<\cdots,
	$$
	with the following property: 
	\begin{align}\label{property}
	\overline{f}^{p,r}_k\le C_r \cdot \underline{f}^{p,r}_k, \quad \forall k\in\cI,~ \forall p\in[P_r],~ \forall r\in[d],
	\end{align}
	where $\underline{b}^{r}:=\min_{e\in\cE,~t\in[T_r]:~b^{r,t}(e)>0} b^{r,t}(e)$, $\overline{b}^{r}:=\max_{e\in\cE,~t\in[T_r]} b^{r,t}(e)=$ $b^{r,t}$ $(e_n), \underline{f}^{p,r}_k:=\min_{e\in\cE^r_p:~f_k^r(e)>0}f_k^r(e)$  and $\overline{f}^{p,r}_k:=\max_{e\in\cE^r_p}f_k^r(e)=f_k^r(e_{\overline{i}(p,r)})$.
\end{lemma}

\begin{proof}
	Fix $r\in[d]$. For $t\in[T_r]$, 
	let $j^{t,1}\in\cV$ be the smallest index such that ${b}^{r,t}((j^{t,1},j^{t,1}+1))>0$, and for $\ell'=2,3,\ldots,$ let $j^{t,\ell'}\in\cV$, be the smallest index such that 
	\begin{align}\label{eeee1}
	b^{r,t}((j^{t,\ell'},j^{t,\ell'}+1)) &> C_r\cdot b^{r,t}((j^{t,\ell'-1},j^{t,\ell'-1}+1)).
	\end{align} 
	{Let $\bar{\ell}$ be the largest index} for which~\raf{eeee1} is possible (if no such index exists, then the lemma follows with $P_r=1$), and  {set $\ell_{t}:=\bar{\ell}+1$ and $j^{t,\ell_t}:=m$.} The inequality in~\raf{eeee1} implies that  
	$b^{r,t}(e_n)>C_r^{\ell_t-1}\cdot b^{r,t}((j^{t,1},j^{t,1}+1))$ which implies in turn that 
	\begin{align}\label{b1}
	\ell_t\le  \log_{C_r}\frac{b^{r,t}(e_n)}{ b^{r,t}((j^{t,1},j^{t,1}+1))}\le \log_{C_r} \Big( \frac{\overline b^r}{\underline b^r}\Big).
	\end{align}
	Moreover, \raf{eeee1} implies
	\begin{align}\label{fff1}
	\frac{{b}^{r,t}((j-1,j))}{{b}^{r,t}((j'-1,j'))}\le C_r,\quad\forall j,j'\in\{j^{t,\ell'}+1,\ldots,j^{t,\ell'+1}\},\forall \ell'=2,\ldots,\ell_t-1. 
	\end{align}
	
	The set $\bigcup_{t\in[T_r]}\{j^{t,\ell'}:\ell'\in[\ell_t]\}\subseteq\cV$ defines a partition of $\cE$ into $P_r\le\sum_{t=1}^{T_r}(\ell_t-1)$ intervals $\cE^r_1,\ldots,\cE^r_{P_r}$. By \raf{b1}, 
	\begin{align}
	\label{b1--}
	P_r< T_r\log_{C_r} \Big( \frac{\overline b^r}{\underline b^r}\Big).
	\end{align} 
	Consider any interval $\cE^r_p:=\big\{e_{\underline{i}(1,p)},e_{\underline{i}(1,p)+1},\ldots,e_{\overline{i}(1,p)}\big\}$ in the partition. Then by~\raf{funct} and~\raf{fff1}, for any $e,{e}'\in\cE^r_p$, we have
	$$
	\tilde{b}^{r,t}(e)\le C_r\cdot \tilde{b}^{r,t}(e'),\quad \text{ whenever }e'\ge e^r_k
	$$
	and thus, it follows form~\raf{funct} that, whenever  $f_k^r(e')>0$ (and hence $e'\ge e^r_k$), we have
	\begin{align*}
	f_k^r(e)&=\sum_{t=1}^{T_r}a_k^{r,t}\tilde b^{r,t}_k(e)\le \sum_{t=1}^{T_r}a_k^{r,t} C_r \tilde b^{r,t}_k(e')
	\le C_r f_k^r(e'),
	\end{align*} 
	as required by~\raf{property}. \hfill$\blacksquare$
\end{proof} 
\section{A QPTAS for separable $d$-USFP}\label{sec:qpt}

This section presents an LP-based approach that arrives at a QPTAS for separable $d$-USFP with the main result stated in Theorem~\ref{th:mcqptas-}. {The high-level idea behind the provided scheme is to segment the users' demand functions in each partition of edges guaranteed by Lemma~\ref{l5} into ``large'' and ``small'', then effectively combine their solutions by exploiting monotonicity and separability of these functions. As the number of ``large'' demands in the optimal solution turns to be provably bounded, we guess the corresponding decision variables through exhaustive search. On the other hand, the situation with ``small'' demands is more complicated since their presence in the optimal solution can be significant. However, as shown in Lemma~\ref{l4}, for such demands, a given fractional solution $\tilde x$ for separable $d$-USFP can be rounded to an integral one that fits within $\tilde x$'s resource requirements without a notable sacrifice in the objective value.}

For exposition clarity, the analysis is arranged into two subsections, which are then further dissected into more concise paragraphs. We proceed by exploring the properties of near-optimal solutions.



\subsection{Structure of Near-optimal Solutions}\label{sec:structure}
\paragraph{\textbf{Discretizing the instance:}}
Let $u_{\max}:=\max_{k\in\cI}u_k$ and $\epsilon\in(0,1)$ be a given constant. Define $\hat\cI:=\{k\in\cI:~u_k\ge\frac{\epsilon u_{\max}}{n}\}$. Note that $u_{\max}\le\OPT$ for a feasible instance, where $\OPT$ is the value of an optimal solution for $d$-USFP$[\cI,c]$. It follows that $\sum_{k \in\cI\setminus\hat\cI}u_k\le \epsilon u_{\max}\le \epsilon\OPT$ and hence, $\sum_{k \in\hat\cI}u_k\ge (1-\epsilon)\OPT$. 

For $k\in\hat\cI$ and $r\in[d]$, let $\underline{f}^r_k:=\min_{e:~f_k^r(e)>0}f_k^r(e)$,  $\overline{f}^r_k:=\max_{e}f_k^r(e)=f_k^r(e_n)$, $\underline{f}^r:=\min_{k} \underline f_k^r$ and $\overline{f}^r:=\max_{k} \overline f_k^r$. We consider discrete levels of function values: for $l=-\infty,0,1,2,\ldots,\left\lceil\log_{(1+\epsilon)}\frac{n\overline f^r}{\underline f^r}\right\rceil$ let $F_l^r:=(1+\epsilon)^l\underline f^r$,  and $F^r:=\Big\{F_l^r:~l=-\infty,0,1,2,\ldots,\left\lceil\log_{(1+\epsilon)}\frac{n\overline f^r}{\underline f^r}\right\rceil\Big\}$ {with $\overline F:=\max\Big\{|F^r| : r\in[d]\Big\}$.}

\paragraph{\textbf{Partitioning the instance:}} For each $r\in[d]$, we assume the partition of $\cE$ guaranteed by Lemma~\ref{l5},  and let $\underline {a}^r:=\min_{k\in\hat \cI,~t\in[T_r]:~a_k^{r,t}>0}a_k^{r,t}$ and $\overline {a}^r:=\max_{k\in\hat \cI,~t\in[T_r]}$ $a_k^{r,t}$. Note that if $a_k^{r,t}>0$ and $k\in\hat\cI$, then
$
\frac{\epsilon u_{\max}}{n\overline a^r}\le\frac{u_k}{a_k^{r,t}}\le \frac{u_{\max}}{\underline{a}^r}.
$
We partition the users in $\hat\cI$ into $Q:=\prod_{r=1}^d\prod_{t=1}^{T_r}Q_{r,t}$ groups, where $Q_{r,t}:=\left\lceil\log\frac{n\overline a^r}{\epsilon\underline a^r}\right\rceil+1$:
\begin{align}\label{grouping}
\cI^{q}=\Big\{k\in\hat\cI:~2^{q_{r,t}-1}L\le\frac{u_k}{a_k^{r,t}}< 2^{q_{r,t}}L\text{ for all }t\in[T_r],~ r\in[d]\Big\},
\end{align}
 for $q=(q_{r,t}:~t\in[T_r],~r\in[d])\in \cQ:=\prod_{r=1}^d\prod_{t=1}^{T_r}\{1,\ldots,Q_{t,r}-1,\infty\}$, where\footnote{For clarity, it is assumed in \raf{grouping} that the strict inequality is replaced by an inequality when $a_k^{r,t}=0$.} $L:=\frac{\epsilon u_{\max}}{n\overline a^r}$.  Let $\overline Q:=\max_{t,r}Q_{r,t},$. 
 Then $Q\leq \overline Q^{\sum_{r=1}^d{T_r}}$.

\paragraph{\textbf{Structure of the optimal solution:}} Consider an optimal solution $x^*$ to separable $d$-USFP[$\cI,c$]. For $q\in\cQ$, let $\cT^*=\{k\in\hat\cI:~x_k^*=1\}$. Then $(f^*)^{q,r}(e):=\sum_{k\in\cT^*\cap\cI^{q}}f_k^r(e)$, for $r\in[d]$, defines a monotone non-decreasing function on $\cE$. We call such a function a ``profile" defined by the optimal solution in group $\cI^q$. For $p\in[P_r]$, let $(h^*)^{q,p,r}=\max_{e\in\cE^r_p}(f^*)^{q,r}(e)$ be the peak demand defined by the optimal solution (from group $q$) within the interval $\cE_p^r$. 

For $q\in \cQ$, let $(\cL^*)^q:=\{k\in \cI^q\cap\cT^*:~\underline{f}_k^{p,r}>\epsilon^2(h^*)^{q,p,r}\text{  for some } p\in[P_r],~ r\in[d]\}$ be the set of ``large" demands within group $\cI^{q}$ in the optimal solution, and let $\cS^q:=\cI^q\cap\cT^*\setminus(\cL^*)^q$ be the set of ``small" demands within the same group. Note that, by definition of $(h^*)^{q,p,r}$ and the monotonicity of $f_k^r(\cdot)$,  there cannot be more than $\frac{1}{\epsilon^2}$ demands $k$ in $\cI^q\cap \cT^*$ such that $\underline{f}^{p,r}_k>\epsilon^2(h^*)^{q,p,r}$, and hence $|(\cL^*)^q|\le \frac{\sum_{r=1}^dP_r}{\epsilon^2}$. 
The situation with small demands is more complicated as their number in the optimal solution can be high. However, with a small loss in the objective value, the profile defined by such small demands can be restricted  into one that admits a small description. This motivates the following definition (generalizing that of in \citep{BCES06}). 
\begin{definition}\label{d3.3}\emph{($(h,\epsilon)$-restricted profile)}
	Let $\epsilon>0$ be such that $1/\epsilon\in\ZZ_+$.
	For $r\in [d]$ and $p\in[P_r]$, let $h=(h^{p,r})_{p\in[P_r],~r\in[d]}$ be a given vector of numbers such that $h^{p,r}\in F^r$ and $h^{p,r}\ge h^{p-1,r}$, for all $p=2,\ldots,P_r$ and $r\in[d]$. 
	An $(h,\epsilon)$-restricted profile $g=(g^r)_{r\in[d]}$ is vector of monotone functions $g^{r}:\cE\to\RR_+$ such that 
$g^r(e)\in\{l\epsilon h^{p,r}:~l\in\{0,1\ldots,1/\epsilon\},~p\in[P_r]\}$ (see Figure~\ref{f1} in Section~\ref{lemma2b} for pictorial interpretation of an $(h,\epsilon)$-restricted profile).
\end{definition}
Accordingly, the total number of $(h,\epsilon)$-restricted profiles is at most $m^{\sum_{r=1}^dP_r/\epsilon}$. For $q\in\cQ$ and for $p\in[P_r]$, define
\begin{align}\label{hpr}
H^{q,p,r}:=\sum_{t=1}^{T_r}\frac{b^{r,t}\big(e_{\overline i(p,r)}\big)}{2^{q_{r,t}}L}.
\end{align}
Note that 
\begin{align}
\forall p\in[P_r]:~H^{q,p,r}>0~&\Leftrightarrow~\exists t\in[T_r]:~q_{r,t}\ne\infty \nonumber\\
&\Leftrightarrow~ \forall k\in\cI^q~\exists t\in[T_r]:~a_k^{r,t}>0~\nonumber\\
&\Leftrightarrow~ \forall k\in\cI^q:~f_k^r(e_n)>0. \label{p1}
\end{align}

Let $\cH^q:=\{r\in[d]:~H^{q,P_r,r}>0\},$ and $\alpha:=\frac{\sum_{r=1}^dP_r}{\sum_{r\in\cH^{q}}P_r}.$ Assume $\cH^q\neq \emptyset$ since otherwise, $f_k^r(e_n)=0$ for all $k\in\cI^q$ and hence all the users in $\cI^q$ can be taken in the solution without affecting the constraints.

In proving Theorem~\ref{th:mcqptas-}, we shall resort to the below Lemma, which builds on top of the findings in~\citep{BCES06} and is proved in Section~\ref{lemma2b}.

\begin{lemma}\label{l4}
	Fix $q\in\cQ$ and $\epsilon\in(0,1)$.
	Let $\cS^q\subseteq\cI^q$ be  a set of demands within group $q$ such that $\underline{f}_k^{p,r}\le B^{p,q,r}$ for all $k\in\cS^q$, $p\in[P_r]$, $r\in[d]$, and some numbers $B^{p,q,r}\in\RR_+$. Let $h^q=(h^{q,p,r})_{p\in[P_r],~r\in[d]}$ be a given vector of numbers such that $h^{q,p,r}\in F^r$ and $h^{q,p,r}\ge h^{q,p-1,r}$, for all $p=2,\ldots,P_r$ and $r\in[d]$, {and $(\tilde{x}_k)_{k\in\cS^q}\in[0,1]^{\cS^q}$ be such that}
	\begin{align}	\label{rel0}
	\sum_{k\in\cS^q}\overline f_k^{p,r}\tilde x_k\le(1+\epsilon)h^{q,p,r},\quad\forall p\in[P_r],~\forall r\in[d].
	\end{align}
	 Then we can find in polynomial time an integral vector $(\hat x_k)_{k\in\cS^q}\in\{0,1\}^{\cS^q}$ and {an $(h,\epsilon)$-restricted} profile $g^{q}$, such that
	 ~\\
	
	\noindent (i) $\sum_{k\in\cS^q}f_k^r(e)\hat x_k\leq g^{q,r}(e)\leq\sum_{k\in\cS^q}f_k^r(e)\tilde{x}_k$ for all $e\in\cE,~r\in[d]$, and 
	
	\noindent (ii) $\sum\limits_{k\in\cS^q}u_k\hat x_k\geq  \sum\limits_{k\in\cS^q}u_k\tilde x_k-\sum\limits_{r\in\cH^q}\left(\sum\limits_{p=1}^{P_r}\left(\frac{C_r}{H^{q,p,r}}\left(\epsilon h^{q,p,r}+B^{q,p,r}\right)\right) + \frac{\alpha P_rB^{q,P_r,r}}{\epsilon H^{q,P_r,r}}\right).$ 
\end{lemma}
In other terms, Lemma~\ref{l4} establishes that, when all demands are small, a given fractional solution $\tilde x$ for separable $d$-USFP can be rounded to an integral solution $\hat x$ that fits within a capacity profile with a small description, losing only a small part of the utility of $\tilde x$.

\subsection{Approximation Scheme} 

The featured QPTAS, formally stated in Alg.~\ref{alg:qptas}, proceeds as follows. As $\OPT\ge u_{\max}$, by restricting the set of demands to  $\hat \cI$ (defined in Section~\ref{sec:structure}) we lose only a value of at most $\epsilon\OPT$ from the optimal solution. Next, the algorithm discretizes the instance and partitions the users in $\hat\cI$ into $Q$ groups $(\cI^q)_{q\in\cQ}$, as described in Section~\ref{sec:structure}. Additionally, Alg.~\ref{alg:qptas} partitions $\cE$ into intervals $\cE^r$ satisfying  assumption~\raf{property}, as per Lemma~\ref{l5} (with $C_r=2$). 

\begin{algorithm}[!htb]
	\caption{ {\sc $d$-USFP-QPTAS}  }
	\label{alg:qptas}
	\begin{algorithmic}[1]
		\Require An approximation parameter $\epsilon\in(0,1)$; separable {\sc $d$-USFP} input $(f_k^{r})_{k\in \cI,~r\in[d]}$ satsifying~\raf{funct}; {capacities $(c^r)_{r\in[d]}$}
		\Ensure	An integral solution $\hat x$ to $d$-USFP such that $\sum_{k\in\cI}u_k\hat x_k \ge (1-O(\epsilon)) \OPT$
		
		\For {{each selection $\Big(\cL = (\cL^q)_{q\in\cQ}, h = \big(h^q=(h^{q,p,r})_{p\in[P_r],~r\in[d]}\big)_{q\in\cQ}\Big)$ such that} $\cL^q\subseteq \cI$, $|\cL^q| \le \frac{\sum_{r=1}^dP_r}{\epsilon^2}$  and $h^{q,p,r}\in F^r$ } \label{alg:guess} 
		\If{{$\sum_{k\in\cL}f_k^r(e)$}$+\sum_{p\in[P_r],~q\in\cQ}h^{q,p,r}\le c^r(e)$ $\forall e\in\cE,~r\in[d]$} \label{qptas:mc-feas}
		\State {$\hat x_k'\gets1$ $\forall~ k\in \cL$}
		\For{$q\in\cQ$}\label{qptas:Q}
		\State Let $\cS^q$ be given by~\raf{small}
		\For {every $(h,\epsilon)$-restricted profile $g^q$}\label{qptas:prof}
		\State$ (\hat x_k')_{k\in \cS^q} \leftarrow $ Integral vector returned by applying Lemma~\ref{l4}  with vector $h^q$, {and $(\tilde x_k)_{k\in \cS^q}=(x'_k)_{k\in \cS^q}$} \label{qptas:mc-round}
		\EndFor
		\EndFor
		\If{$\sum_{k\in\cI}u_k\hat x_k' > \sum_{k\in\cI}u_k\hat x_k$ }
		\State $ \hat x\leftarrow \hat x'$
		\EndIf 
		\EndIf
		
		\EndFor
		\State \Return $\hat x$
	\end{algorithmic}
\end{algorithm}

Then for each group $q\in\cQ$, Alg.~\ref{alg:qptas} guesses the set of large demands $\cL^q\subseteq\cI^q$ in the optimal solution, and the peaks $h^{q,p,r}$, within $1+\epsilon$, of the small demands in the optimal solution within the interval $\cE_p^r$. Let $\cL=(\cL^q)_{q\in\cQ}$ and $h^q=(h^{q,p,r})_{p\in[P_r],~r\in[d]}$ where $h^{q,p,r}\in F^r$. Define the set of small demands within group $q\in\cQ$ as
\begin{align}\label{small}
\cS^q:=\left\{k\in\cI^q:~\underline{f}_k^{p,r}\le B^{q,p,r}\text{ for all } p\in[P_r],~r\in[d]\right\},
\end{align}
where $B^{q,p,r}:=\epsilon^2\left[h^{q,p,r}+\sum_{k\in\cL^q}\overline{f}_k^{p,r}\right]$.

Let $T:=\max_{r}T_r$ and $M:=\max\left\{\max_{r}\frac{\overline a^r}{\underline a^r},\max_{r}\frac{\overline b^r}{\underline b^r}\right\}.$
\begin{theorem}
	\label{th:mcqptas-}
	For any fixed $\eps\in(0,1)$, Alg.~\ref{alg:qptas} attains a $(1- \eps)$-approximation for separable $d$-USFP in time $\left(\frac{nm\log (dnTM)}{\eps}\right)^{dT\cdot O\left(\log  \frac{dnTM}{\eps}\right)^{dT}}$. 
\end{theorem}

\begin{proof}
	Let $\epsilon:=\frac{\eps}{2\beta+1}$, where $\beta=\max_{r\in\cH^q}2\left(2C_r+\alpha P_r\right)=O(d^3(T\log M)^2)$. 
	The number of possible choices for each $\cL^q$ in step~\ref{alg:guess} of Alg.~\ref{alg:qptas} is at most $n^{\sum_{r=1}^dP_r/\epsilon^2}$. Thus, using $Q\leq \overline Q^{\sum_{r=1}^dT_r}$, and $\overline Q=O(\log\frac{nM}{\epsilon})$, the number of possible choices for $\cL$ is at most \begin{align}\label{ch-L-}
	n^{\sum_{r=1}^dP_rQ/\epsilon^2}\le n^{\sum_{r=1}^dP_r\overline Q^{{\sum_{r=1}^dT_r}}/\epsilon^2}= n^{dT\log M\cdot O(\log  \frac{nM}{\epsilon})^{dT/\epsilon^2}}.
	\end{align} The number of choices for each $h^q=(h^{q,p,r})_{p\in[P_r],~r\in[d]}$ is $$\overline F^{\sum_{r=1}^dP_r}= O\Big(\big(\frac{\log(nTM)}{\epsilon}\big)^{dT\log M}\Big)\,,$$ and the number of choices for $Q$ in step~\ref{qptas:Q} is 
	\begin{align}\label{ch-Q-}
	\overline{Q}^{\sum_{r=1}^d {T}_r}\le O\left(\log  \frac{nM}{\epsilon}\right)^{dT},
	\end{align}
	giving at most
	\begin{align}
	\label{ch-h-}
	\left(O\left(\frac{\log(nTM)}{\epsilon}\right)^{dT\log M}\right)^{ Q}
	=\left(O\left(\frac{\log(nTM)}{\epsilon}\right)^{dT\log M}\right)^{O\left(\log  \frac{nM}{\epsilon}\right)^{dT}}
	\end{align} choices for $h=(h^q)_{q\in\cQ}$ in step~\ref{alg:guess}.
	The number of choices for the $\epsilon$-restricted profiles in step~\ref{qptas:prof} is bounded from above by
	$
	m^{\sum_{r=1}^dP_r/\epsilon}=m^{O(\frac{dT\log{M}}{x\epsilon})}.
	$	 The bound on the running time of Alg.~\ref{alg:qptas} follows from this and~\raf{ch-L-},\raf{ch-Q-},\raf{ch-h-}. 
	
	We now argue that the solution $\hat x$ outputted by Alg.~\ref{alg:qptas} is $(1-O(\epsilon))$-approximation for separable $d$-USFP. 
	Let $x^*$ be an optimal solution for {\sc $d$-USFP} of objective value $\OPT \triangleq \sum_{k\in\cI}u_kx^*_k$. By the definition of $\hat\cI$, we have
	\begin{align}\label{en0-}
	\sum_{k\in\cI\setminus\hat\cI}u_k\le\epsilon\OPT.	
	\end{align}
	Define  $\cT^*\triangleq\{k \in \hat\cI \mid x^*_k = 1\}$ and $(h^*)^{q,p,r}=\sum_{k\in\cT^*\cap\cI^q}\overline f_k^{p,r}$, for $p\in[P_r],$ $r\in[d]$ and $q\in\cQ$.  
	Let $(\cL^*)^q:=\Big\{k\in \cI^q\cap\cT^*:~\underline{f}_k^{p,r}>\epsilon^2(h^*)^{q,p,r}\text{  for some } p\in[P_r],~\text{and some } r\in[d]\Big\}$ be the set of ``large'' demands within group $\cI^{q}$ in the optimal solution, and let $(\cS^*)^q:=\cI^q\cap\cT^*\setminus(\cL^*)^q$ be the set of ``small" demands within the same group. Note by this definition that $|(\cL^*)^q|\le\frac{\sum_{r=1}^dP_r}{\epsilon^2}$, and thus $\cL^*=((\cL^*)^q)_{q\in\cQ}$ and $h=(h^q)_{q\in\cQ}$ will be one of the guesses considered by the algorithm in step~\ref{alg:guess}. Let us focus on this particular iteration of the loop in step~\ref{alg:guess}. {Let $h^{q,p,r}=(1+\epsilon)^{\underline{\ell}}\underline{f}^{r}$, where $\underline{\ell}$ is the smallest integer (including $-\infty$) such that  $h^{q,p,r}+\sum_{k\in(\cL^*)^q}\overline f_k^{p,r}\ge  (h^*)^{q,p,r}$.} Note that $h^{q,p,r}\in F^r$,  and
	\begin{align}\label{rel-}
	\frac{1}{1+\epsilon}h^{q,p,r}+\sum_{k\in(\cL^*)^q}\overline f_k^{p,r}\le (h^*)^{q,p,r}\le h^{q,p,r}+\sum_{k\in(\cL^*)^q}\overline f_k^{p,r}.
	\end{align}
	
	Note that for any $k\in(\cS^*)^q$, $q\in\cQ$, $p\in[P_r]$, and $r\in[d]$, we have by \raf{rel-},
	$$
	\underline f_k^{p,r}\le\epsilon^2 (h^*)^{q,p,r}\leq\epsilon^2\left(h^{q,p,r}+\sum_{k\in(\cL^*)^q}\overline f_k^{p,r}\right),
	$$
	and hence $(\cS^*)^q\subseteq\cS^q$. Note also that 
	\begin{align}
	\label{bbb2-}
	B^{q,p,r}&=\epsilon^2\left[h^{q,p,r}+\sum_{k\in(\cL^*)^q}\overline{f}_k^{p,r}\right] \nonumber \\
	&\le\epsilon^2\left[h^{q,p,r}+(1+\epsilon)\sum_{k\in(\cL^*)^q}\overline{f}_k^{p,r}\right]\le \epsilon^2(1+\epsilon)(h^*)^{q,p,r}.
	\end{align}
	For each $q\in\cQ$, there is an $(h,\epsilon)$-restricted profile $g^q$ and an integral solution $ (\hat x_k')_{k\in \cS^q}  $ that satisfy Lemma~\ref{l4} (applied with $\hat x\leftarrow \hat x'$ and $\tilde x\leftarrow x^*$). Since all the possible $(h,\epsilon)$-restricted profiles are probed, the profile $g^q$ will be identified in one of the iterations of the loop in step~\ref{qptas:prof} of Alg.~\ref{alg:qptas}. Let us consider this iteration.
	By condition~(ii) of Lemma~\ref{l4} and \raf{bbb2-}, 
	{\footnotesize
		\begin{align}\label{en1-}
		\sum_{k\in\cS^q}u_k\hat x_k'&\geq  \sum_{k\in\cS^q}u_kx_k^*-\sum_{r\in\cH^q}\left(\sum_{p=1}^{P_r}\left(\frac{C_r\left(\epsilon h^{q,p,r}+B^{q,p,r}\right)}{H^{q,p,r}}\right)+\frac{\alpha P_rB^{q,P_r,r}}{\epsilon H^{q,P_r,r}}\right)\nonumber\\
		&=\sum_{k\in\cS^q}u_kx^*_k- \sum_{r\in\cH^q}\left(\sum_{p=1}^{P_r}\left(\frac{C_r\epsilon (1+\epsilon)^2(h^*)^{q,p,r}}{H^{q,p,r}}\right)+\frac{\alpha P_r\epsilon^2(1+\epsilon)(h^*)^{q,P_r,r}}{\epsilon H^{q,P_r,r}}\right)\nonumber\\
		&=\sum_{k\in\cS^q}u_kx^*_k-\epsilon(1+\epsilon)\sum_{r\in\cH^q}\left(\sum_{p=1}^{P_r}\left(\frac{C_r(1+\epsilon)(h^*)^{q,p,r}}{H^{q,p,r}}\right)+\frac{\alpha P_r(h^*)^{q,P_r,r}}{H^{q,P_r,r}}\right).	
		\end{align}
	}%
	On the other hand, for $k\in\cS^q$ and $r\in[d]$ such that $f_k^r(e_n)>0$ (and hence $H^{q,p,r}>0$ for all $p\in[P_r]$ by \raf{p1}), we have $u_k\ge 2^{q_{r,t}-1}La_k^{r,t}$  and thus 
	\begin{align}
	\label{bb2-}
	u_k\frac{b^{r,t}(e_{\overline{i}(p,r)})}{2^{q_{r,t}-1}L}\ge a_k^{r,t} b^{r,t}(e_{\overline{i}(p,r)})\ge a_k^{r,t} \tilde b^{r,t}(e_{\overline{i}(p,r)}).
	\end{align}
	Summing up \raf{bb2-} over $t\in[T_r]$, we get $u_k\ge \frac{\overline f_k^{p,r}}{2H^{q,p,r}}$. {Recall that $(h^*)^{q,p,r}=\sum_{k\in\cT^*\cap\cI^q}\overline f_k^{p,r}$,} then summing this inequality over $k\in\cT^*\cap\cI^q$ yields
	\begin{align}
	\label{LB-}
	\OPT^q:=\sum_{k\in\cT^*\cap\cI^q}u_k\ge \sum_{k\in\cT^*\cap\cI^q}\frac{\overline f_k^{p,r}}{2H^{q,p,r}}=\frac{(h^*)^{q,p,r}}{2H^{q,p,r}}.
	\end{align} 
	Summing~\raf{LB-}, over $r\in\cH^q$ and $p\in[P_r]$ gives
	{
	\begin{align}
	\label{LA2-}
	\OPT^q&\ge\sum_{r\in\cH^q}\sum_{p=1}^{P_r}\frac{(h^*)^{q,p,r}}{2H^{q,p,r}} \nonumber \\
	&=\frac1{\beta}\cdot\sum_{r\in\cH^q}\sum_{p=1}^{P_r}\frac{\beta(h^*)^{q,p,r}}{2H^{q,p,r}} \nonumber \\ 
	&\geq\frac1{\beta}\cdot\sum_{r\in\cH^q}\sum_{p=1}^{P_r}\frac{2\left(2C_r+\alpha P_r\right)(h^*)^{q,p,r}}{2H^{q,p,r}} \nonumber \\ 
	&\geq\frac1{\beta}\cdot\sum_{r\in\cH^q}\sum_{p=1}^{P_r}\Bigg(\frac{2C_r(h^*)^{q,p,r}}{H^{q,p,r}} + \frac{\alpha P_r(h^*)^{q,p,r}}{H^{q,p,r}} \Bigg)\nonumber \\ 
	&\geq\frac1{\beta}\cdot\sum_{r\in\cH^q}\Bigg(\sum_{p=1}^{P_r}\frac{(1+\epsilon)C_r(h^*)^{q,p,r}}{H^{q,p,r}} + \sum_{p=1}^{P_r}\frac{\alpha P_r(h^*)^{q,p,r}}{H^{q,p,r}} \Bigg)\nonumber \\ 
	&\ge\frac1{\beta}\cdot\sum_{r\in\cH^q}\left(\sum_{p=1}^{P_r}\left(\frac{C_r}{H^{q,p,r}} (1+\epsilon)(h^*)^{q,p,r}\right)+\frac{\alpha P_r(h^*)^{q,P_r,r}}{H^{q,P_r,r}}\right)\,,
	\end{align}}\noindent{where $\beta=\max_{r\in\cH^q}2\left(2C_r+\alpha P_r\right)$ as defined previously.}
	Thus, it follows from~\raf{en1-} and~\raf{LA2-} that 
	\begin{align}\label{en2-}
	\sum_{k\in\cS^q}u_k\hat x_k'&\geq  \sum_{k\in\cS^q}u_kx_k^*-\epsilon(1+\epsilon)\beta\OPT^q\geq  \sum_{k\in(\cS^*)^q}u_kx_k^*-\epsilon(1+\epsilon)\beta\OPT^q. 
	\end{align}
	Summing~\raf{en2-} over all $q\in\cQ$ and using~\raf{en0-} and \raf{en2-} gives
	\begin{align*}
	\sum_{k\in\cI}u_k\hat x_k'&=\sum_{q\in\cQ}\left(
	\sum_{k\in(\cL^*)^q}u_k\hat x_k'+\sum_{k\in\cS^q}u_k\hat x_k'\right)\\
	&\ge \sum_{q\in\cQ}\left(
	\sum_{k\in(\cL^*)^q}u_kx_k^*+\sum_{k\in(\cS^*)^q}u_kx_k^*-\epsilon(1+\epsilon)\beta\OPT^q\right)\\
	&=\sum_{k \in\cT^*}u_kx_k^*-\epsilon(1+\epsilon)\beta\sum_{k\in\cT^*}u_k\\
	&=\sum_{k \in \hat\cI}u_kx_k^*-\epsilon(1+\epsilon)\beta\sum_{k\in\cT^*}u_k\\
	&\ge \sum_{k \in \cI}u_kx_k^*-\epsilon(2\beta+1)\OPT= (1-\eps)\OPT.
	\end{align*}
	It follows that the solution $\hat x$ returned by Alg.~\ref{alg:qptas} satisfies $$\sum_{k\in\cI}u_k\hat x_k\ge\sum_{k\in\cI}u_k\hat x_k'\ge(1-\eps)\OPT\,,$$
	thus concluding the proof. \hfill$\blacksquare$
\end{proof}

Note that the running time is quasi-polynomial if $M=2^{\polylog (m,n)}$ and $d=O(1)$, $T=O(1)$.

\section{A Logarithmic approximation for single-source separable $d$-USFP}\label{sec:logarithmic}

Notwithstanding its theoretical appeal, the QPTAS devised in Sec.~\ref{sec:qpt} is computationally prohibitive even for modest problem sizes, hence is of limited practicality. This section presents an efficient logarithmic approximation for single-source separable $d$-USFP$[\cI,c]$ with a running time complexity dominated by solving an LP. Before stating the result formally, we rewrite the problem in a suitable matrix notation and briefly outline the underlying technique. Notice that $d$-USFP$[\cI,c]$ can be cast as a general {\it packing integer program} (PIP) of the form
\begin{align}
\Big(\mathcal{P}\big[(A^r)_{r \in [d]}, u, (c^r)_{r \in [d]}\big]\Big) & \quad \max_{\substack{x}} ~~ u^Tx  &\qquad \qquad\notag \\
\text{s.t.} \ & \quad A^rx\le c^r,~\quad \forall~r\in[d]\label{lg0} &\\
& \quad x \in \{0,1\}^n,\label{lg1}
\end{align}
where $u \in \RR_+^n$ is the utility vector, $c^r \in \RR_+^m$ denotes the edge capacities in dimension $r \in [d]$ and  $A^r \in \RR_+^{m \times n}$ resembles the edge-demand incidence relation for the corresponding dimension $r \in [d]$, with the rows signifying the edges and the columns the demands (i.e., $A_{ik}^r = f_k^r (e_i)$ for $\forall i \in [m], k \in \cI$).

Exploiting the special structure of $\mathcal{P}$ induced by the monotonicity and separability of demands, we develop a simple {\it grouping and scaling} method allowing to reduce the problem to an easier instance with only logarithmically many constraints. Recall that an analogously named technique was derived in~\citep{kolliopoulos2001approximation} for the single-source unsplittable flow problem. Deviating from the setting in~\citep{kolliopoulos2001approximation} of partitioning the instance in the demand space, the proposed approach, instead, decomposes the edges into disjoint segments, each defining a subproblem of $\mathcal{P}$ where each capacity and demand varies within a preset range. These subproblems, after certain alterations, are then reconsolidated, effectively formulating the compacted problem with O$(d\log n)$ number of constraints. It's noteworthy that this reduction subroutine {\it holds irrespective of} the rather restrictive NBA condition, which is stipulated in~\citep{kolliopoulos2001approximation}. Thereafter, invoking the standard {\it randomized rounding} algorithm on the natural LP relaxation of the reduced problem ensures the claimed approximation factor. Formally, the preceding analysis culminates in Theorem~\ref{logap}.

In proving Theorem~\ref{logap}, we capitalize on several established results on randomized rounding and its derandomization (codified in the theorem to follow) as a unified black box technique and thereby omit the intricate particulars.

\begin{theorem}[\citep{Srinivasan1999,Raghavan1987}~] Let $\mathcal{B}$ be a PIP of the form $\max \{ u^Tx : Ax \leq c, x \in \{0,1\}^n\}$, where $A \in [0,1]^{m \times n}$, $u \in [0,1]^n$ and $c \in [1,\infty)^m$ with $\max_j u_j = 1$. Then, there exists an algorithm outputting in deterministic polynomial time a feasible solution to $P$ of value $$\Omega \bigg( \max \bigg\{  \frac{\OPT_L}{m^{1/ \nu}},  \bigg( \frac{\OPT_L}{m^{1/ \nu}} \bigg)^{\frac{\nu}{\nu -1}} \bigg\}\bigg) \, ,$$
where $\OPT_L$ is the optimum of the linear relaxation of $\mathcal{B}$ and $\nu = \min_j c_j$. \label{rnd}
\end{theorem}

\begin{theorem}\label{logproof}There is an O$(d\log n)$-approximation for single-source separable $d$-USFP, provided the edge capacities and demands are bounded by $2^{\polylog(n)}$. \label{logap} 
\end{theorem}

\begin{proof} Let $\Lambda = \Big( (A^r)_{r \in [d]}, u, (c^r)_{r \in [d]} \Big)$ be an input instance of $\mathcal{P}$ with $\OPT$ denoting the value of its optimal solution $x^*$. From $\Lambda$, construct an augmented instance  $\Lambda^{\prime} = \Big( \big(\begin{bmatrix}
	A^r&c^r
	\end{bmatrix}\big)_{r \in [d]}, (u, 0), (c^r)_{r \in [d]} \Big)$, which essentially models the outcome of incorporating {\it a dummy request} with a utility of $0$ and a demand equal to edge capacities. This auxiliary step, meant to streamline the proof, incurs no loss of generality as neither $x^*$ nor its structure is affected in the aftermath. Thus, to elude cumbersome notation, $(A^r)_{r \in [d]}$ and $u$ are hereafter assumed implicitly of the augmented form as in $\Lambda^{\prime}$.
	
	At a loss of only a constant factor in $\OPT$, we shall now transform $\mathcal{P}$ to a problem with O$(d\log n)$ constraints. Let $\Pi$ denote the LP relaxation of $\mathcal{P}$, obtained by allowing $x$ to lie in $[0,1]^n$. Fix a constant $C > 1$, along with the corresponding partitions $(\cE^r)_{r \in [d]}$ guaranteed by Lemma~\ref{l5}, and denote by $A^{r, p}$ the submatrix of $A^r$ restricted to the rows in $\big\{i \in [m] \mid e_i \in \cE_p^r\big\}$. Observe that each interval $\cE^r_p$ in $(\cE^r)_{r \in [d]}$ naturally defines a subproblem $\Pi\big[A^{r,p}, u, c^{r,p}\big]$, where $\frac{\max_{i}A_{ij}^{r,p}}{\min_{i:A_{ij}^{r,p}>0}A_{ij}^{r,p}} \leq C$ for $\forall j \in [n+1]$ and, by introduced ancillary demands, $\frac{\max_{i}c_i^{r,p}}{\min_{i:c_i^{r,p}>0}c_{i}^{r,p}} \leq C$. Given $\Pi\big[A^{r,p}, u, c^{r,p}\big]$, compose a simplified instance $\Pi\big[\overline A^{r,p}, u, \underline c^{r,p}\big]$, with $\underline c^{r,p} := \min_i c^{r,p}_i \cdot \boldsymbol1  $ and $\overline A^{r,p}$ standing for the matrix whose $i,j$-th entry equals $\max_{i}A_{ij}^{r,p}$ if $A_{ij} \neq 0$ and $0$ otherwise. In a sense, this amounts to setting each demand to its maximum, therein flattening out the step functions into lines, and uniforming the edge capacities across the interval. Consider an optimal solution $y^*$ of $\Pi\big[A^{r,p}, u, c^{r,p}\big]$ and set $\tilde {y} := \frac{y^*}{C^2}$. As a corollary, $\tilde{y} $ becomes a feasible solution for $\Pi\big[\overline A^{r,p}, u, \underline c^{r,p}\big]$. On the other hand, any feasible solution to $\Pi\big[\overline A^{r,p}, u, \underline c^{r,p}\big]$ translates into that of $\Pi\big[A^{r,p}, u, c^{r,p}\big]$ of the same value. Taken together and generalized over all the partitions, these observations imply that
	\begin{equation}
	\widetilde {\OPT}_{\Pi} \geq \frac{\OPT_{\Pi} }{C^2} \geq \frac{\OPT }{C^2} \,, \label{eq63}
	\end{equation}
	where $\widetilde{\OPT}_{\Pi}$ and $ \OPT_{\Pi} $ are the optimal objective values of $\Pi\big[(\overline A^r)_{r \in [d]}, u,$ $ (\underline c^r)_{r \in [d]}\big]$ and $\Pi\big[(A^r)_{r \in [d]}, u$ $, (c^r)_{r \in [d]}\big]$, respectively. Furthermore, a finer inspection of the former problem can render the majority of its constraints redundant. Indeed, by construction, each subproblem $\Pi\big[\overline A^{r,p}, u, \underline c^{r,p}\big]$ of $\Pi\big[(\overline A^r)_{r \in [d]}, u, $ $(\underline c^r)_{r \in [d]}\big]$ boils down to a {\it single Knapsack inequality}\footnote{This inequality is captured by the first constraint appearing in the subproblem, and thus can be extracted in O$(1)$ time.} since both, demands and capacities, are levelled therein, and all the requests share the same origin. Compounding these $\tilde{m} = O(d\log n)$ inequalities into $\tilde{A} \in \RR_+^{\tilde{m}\times n+1}$ and $\tilde{c} \in \RR_+^{\tilde{m}\times 1}$, formulate a new PIP $\mathcal{P}\big[\tilde{A} , u, \tilde{c}\big]$ minding that $\widetilde{\OPT}_{\Pi}$  is the optimum value of its linear relaxation.

	Henceforth, it remains to invoke Theorem~\ref{rnd} on $\mathcal{P}\big[\tilde{A} , u, \tilde{c}\big]$ after some proper scaling. In particular, without loss of generality, assume for $\forall i,j$, $\tilde{A}_{i,j} \leq \tilde{c}_i$ since otherwise we might as well set the corresponding $j$-th decision variable to $0$. This being so, scale down each row $i$ of $\tilde{A}$ and $\tilde{c}$ by $\max_j \tilde{A}_{i,j}$, consequently letting $\tilde{A} \in [0,1]^{\tilde{m}\times n+1}$ and $\tilde{c} = \boldsymbol1 $ (due to the dummy requests). Next, scaling $u$ such that $\max_j u_j = 1$,  conforms $\mathcal{P}\big[\tilde{A} , u, \tilde{c}\big]$ to the form in Theorem~\ref{rnd}. Accordingly, we obtain a feasible integral solution to $\mathcal{P}\big[\tilde{A} , u, \tilde{c}\big]$, and hence to $\mathcal{P}\big[(A^r)_{r \in [d]}, u, (c^r)_{r \in [d]}\big]$, of value $ \frac{\widetilde {\OPT}_{\Pi}}{O (d\log n)}$, which together with~\raf{eq63} yields the theorem. \hfill$\blacksquare$
\end{proof}

\vspace*{8pt}
\begin{remark}
	For the sake of variety, the result in this section was provided in an existential form, rather than in an algorithmic variant as in Section~\ref{sec:qpt}. However, the algorithm is straightforward and follows immediately from the proof. Also, it should be noted that, at an additional loss of O$(\log n)$ factor, one can possibly extend this result to separable $d$-USFP through the approach in~\citep{BFKS14} of decomposing the given instance into one in which all the demands intersect.
\end{remark}

\section{From Unsplittable Flows to Electrical Flows: Application to Power Systems}\label{elecflows}

In this section, we develop a reduction procedure that can be applied to LP-based approximations for separable $d$-USFP to produce approximations for AC OPF on line distribution networks. To this end, Section~\ref{oppf} first outlines the pertinent background on OPF and formulates the problem mathematically, then Section~\ref{sec:pre} expounds the proposed reduction.
 



\subsection{AC OPF and its Exact Relaxation for Radial Networks}\label{oppf}

The AC OPF problem, introduced by Carpentier in 1962~\citep{C62}, lies at the heart of techniques routinely deployed in power systems for performance optimization and control (see e.g., \citep{F12a} for a comprehensive survey on OPF). As such, the input of OPF comprises an electrical network, such as the one depicted in Fig.~\ref{f2}, represented by an undirected graph where nodes stand for \textit{electric buses}, whereas the edges model \textit{power lines}. Among the buses, some correspond to AC generators while others to demand nodes (loads). The objective is to determine an operating point, optimal with respect to a given objective (e.g., minimizing generation cost), that satisfies user demands while meeting operational (engineering) constraints (e.g. line thermal limit) and physical properties (imposed by Ohm’s and Kirchoff’s laws) of the electrical network.

From computational perspective, OPF is notoriously toilsome due mainly to the existence of {\em non-convex} constraints involving complex-valued entities of power system parameters such as current, voltage and power. Recently, there has been a major progress on tackling OPF through convex relaxations~\citep{BGCL15,huang2017sufficient,gan2015exact,low2014convex1,low2014convex2}. These papers focus chiefly on {\it radial} (i.e., tree) networks, since they are fairly common in real-world, and derive sufficient conditions under which the convex relaxation is exact (i.e., equivalent to the original non-convex problem); for example, relaxing the rank-$1$ constraint in the semidefinite programming (SDP) formulation \citep{BGCL15}, or relaxing the equality constraints in the second order cone programming (SOCP) formulation \citep{huang2017sufficient,gan2015exact,low2014convex1,low2014convex2}. While these results yield polynomial time algorithms for OPF, their scope is limited to the case with {\it continuously adjustable} power injection constraints; control variables responsible for modulating power loads are fractional and defined in terms of buses). In a more general setting, however, it is often necessary to account for {\it discrete} (or a mix of discrete and continuous) variables~\citep{6629395,7433473,9301221, 8892494}. Specifically, certain loads and devices, e.g., TV, vacuum cleaner or washing machine, operate only under a particular supply of electricity; are either switched on with a fixed power consumption rate or turned off. This \textit{combinatoric structure} renders a substantially more complicated instance of OPF. Concretely, as demonstrated in~\citep{khonji2016optimal}, OPF with discrete demands in a {\it delta} network is hard to approximate within any polynomial guarantees unless P=NP. Prior studies on OPF with discrete control variables, e.g.,~\citep{7926079, 4578739, Hijazi2017}, mainly resort to heuristic techniques, which, per se, are devoid of any optimality guarantees or theoretical guidance.

\begin{figure}
		\centering
		\includegraphics[width=0.6\textwidth]{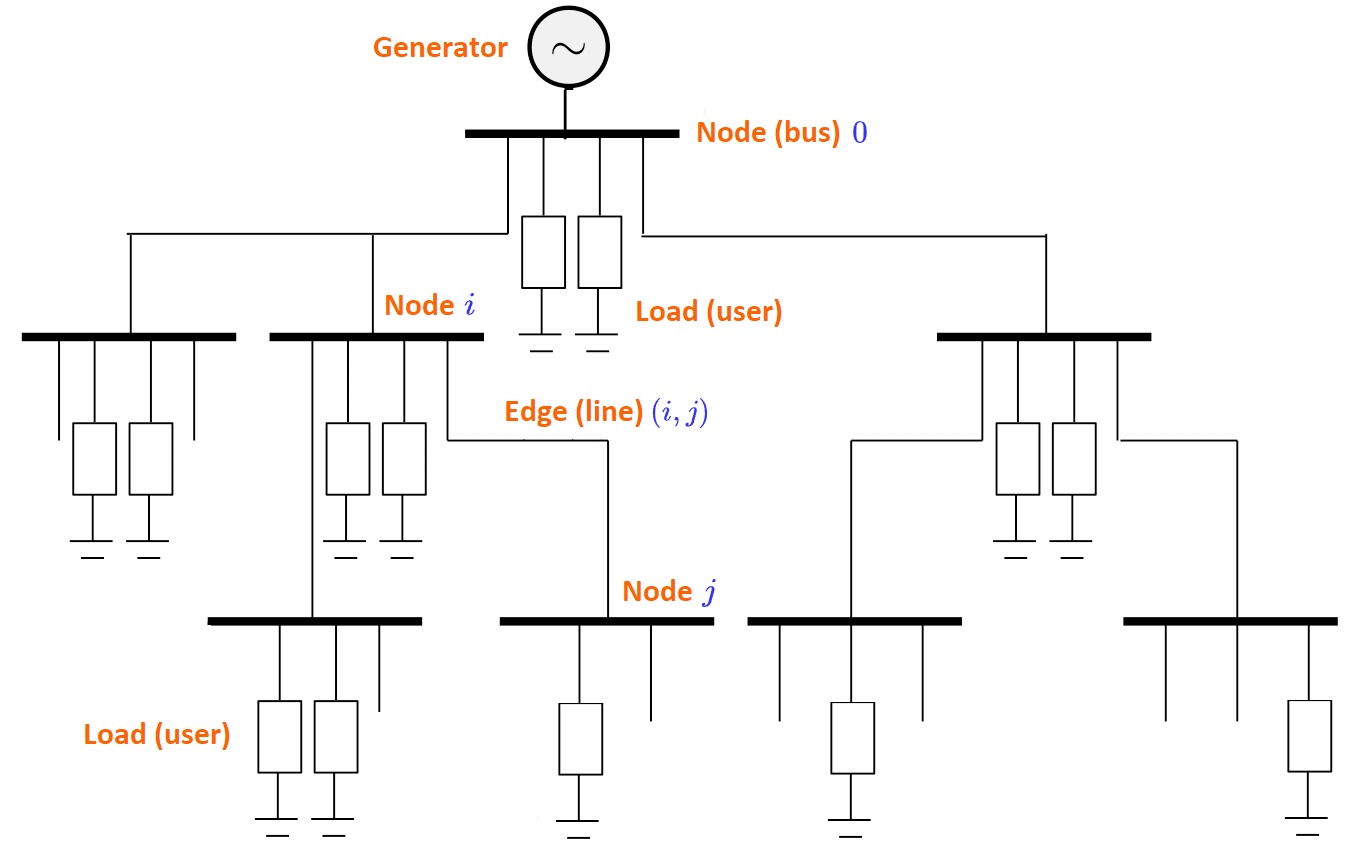}
		\caption{An example of a radial electrical network.}
		\label{f2}
\end{figure}

With the above background in view, we next provide a model of an electrical network and define OPF formally. {Recall from the convention in Sec.~\ref{sec:prelim}  that given a complex number $\nu\in\CC$ we let $|\nu|$ be its {\it magnitude}, $\arg(\nu)$ be the {\it phase angle} that it makes with the real axis, $\nu^*$ be its complex {\it conjugate} and write $\nu^{\rm R} \triangleq {\rm Re}(\nu)$, $\nu^{\rm I} \triangleq {\rm Im}(\nu)$ for its real and imaginary components, respectively.} Consider a radial distribution network represented by a line graph $\cT=(\cV,\cE)$, where $\cV=\{0,1,\ldots,m\}$ denotes the electric {\it buses}, whereas $\cE$ symbolizes the distribution {\it lines}. Each line $e\in\cE$ is characterized by a complex {\it impedance} $z_{e}\in\cC$, with a {\it non-negative} real part resembling the {\it resistance} of the line (to the flow of current) and imaginary part quantifying the  {\it reactance} ({\it inductance} if positive and {\it capacitance} if negative). In the setup under study, a {\it substation generator} is attached to the root of $\cT$, node $0$. By convention, it is assumed that power flows from the root to the nodes. Let $\cV^+\triangleq \cV \setminus\{0\}$ and  $\cV_i^+ \triangleq \cV_i\setminus\{i\}$. When referring to an edge, we shall use the (ordered) pair of subscripts $(i,j)$ and $e$ \textit{interchangeably}, where it is assumed that $i$ is the {\it parent} of $j$ in $\cT$.

At each node $j\in\cV^+$, attached is a set $\cU_j$ of {\it users} (electrical loads). Let $\cN \triangleq \cup_{j\in\cV^+} \cU_j$ be the set of all users ($|\cN|=\tilde{ n}$), while $\cN_j$ be those residing in the subpath rooted at node $j\in\cV^+$. Among these users, some have {\it inelastic} (discrete) power demands, denoted by $\cI \subseteq\cN$. A discrete demand is either completely satisfied or dropped. An example is an appliance that is either switched on with a fixed power consumption rate or  switched off. The rest of users, denoted by $\cF \triangleq \cN\backslash\cI$, have elastic demands which can be partially satisfied. The demand of user $k$ is represented by a complex-valued  number $s_k\in \CC$; the real part $s_k^{\rm R}$ denotes the so-called {\it active} power while the imaginary part $s_k^{\rm I}$ captures the {\it reactive} power; the {\it apparent} power is defined as the magnitude $|s_k|=\sqrt{(s_k^{\rm R})^2+(s_k^{\rm I})^2}$ of $s_k$. Additionally, each user $k\in\cI$ is associated with a number $u_k\in\RR_+$ indicating the {\it utility} of user $k$ if her demand $s_k$ is fully satisfied.

Denote the unique path from node $j$ to the root $0$ by $\cP_j$. For each user $k\in \cU_j$, define $\cP_k \triangleq \cP_j$. With a slight abuse of notation, we interchangeably refer as $\cP_j$  to the set of edges as well as the set nodes on the path from $j$ to the root.
\begin{figure}[!t]
	\centering
	\includegraphics[width=\textwidth]{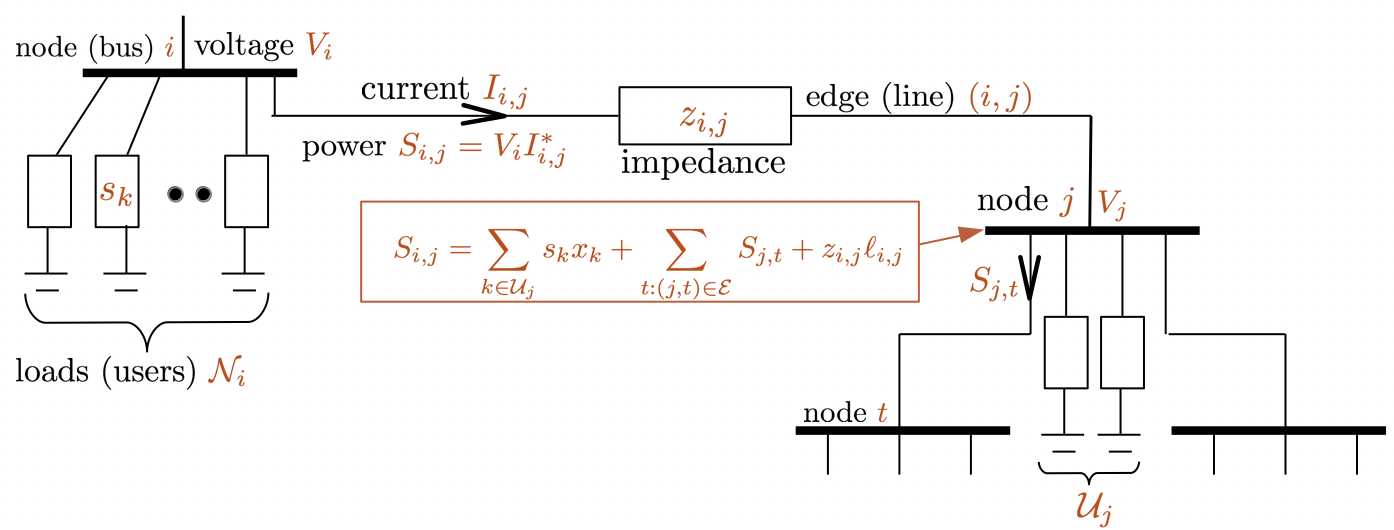}
	\caption{Conservation of power flow  at node $j$.}
	\label{f3}
\end{figure}

A steady-state power flow in a distribution network is generally described by a system of equations. For radial networks (which include paths), these can be framed through the \textit{Branch Flow (a.k.a. DistFlow) Model} (BFM)~\citep{baran19266}. Under BFM, OPF in $\cT$ is embodied by the following mixed-integer programming problem.
\begin{align}
\textsc{Input}:\; \quad & v_0;(\underline v_j, \overline v_j)_{j\in \cV^+}; (\overline S_e, \overline \ell_e, z_e)_{e\in \cE}; (s_k)_{k\in \cN}  \notag\\
\textsc{Output}:\; \quad &s_0 ; (v_j)_{j\in \cV^+}; (S_{e}, \ell_{e})_{e\in \cE}; (x_k)_{k \in \cN} \notag\\
&~\notag \\
\textsc{(OPF)} \quad &\max_{\substack{s_0, x, v,\ell, S \;\;}} f_{\textsc{OPF}}(s_0, x),  \notag \\
\text{s.t.} \ & 	\ell_{i,j} =  \frac{|S_{i,j}|^2}{v_i}, ~~ \forall (i,j) \in \cE  \label{p1:con1} \\
&  S_{i,j}=  \sum_{k \in \cU_j} s_k x_k  + \sum_{t:(j,t)\in \cE} S_{j,t} + z_{i,j}\ell_{i,j}, ~~ \forall (i,j) \in \cE  \label{p1:con2}
\end{align}
\begin{align}
& S_{0,1} =- s_0\label{p1:con3}\\ 
&	v_j = v_i + |z_{i,j}|^2 \ell_{i,j} - 2 \re(z_{i,j}^\ast  S_{i,j}), ~~  \forall (i,j) \in \cE \label{p1:con4} \\
& \underline v_j \le v_j \le \overline v_j,    ~~\forall j \in\cV^+ \label{p1:con5} \\
& |S_e| \le \overline S_e,~|-S_e + z_e\ell_e| \le \overline S_e,  ~~\forall e \in \cE \label{p1:con6}\\
& \ell_e \le\overline \ell_e  ~~\forall e \in \cE \label{p1:con7}\\
& x_k \in \{0,1\},~~\forall k \in \cI,~~   x_k \in [0,1],~~ \forall k \in \cF \label{p1:con8}\\
& v_j \in \RR_+,~\forall  j \in \cV^+ \ \ell_e \in \RR_+,~ S_{e} \in \CC \label{p1:con9},~\forall e\in \cE\,.
\end{align}


\paragraph{\bfseries The variables:} In the above formulation, the complex variable $S_{i,j}$ represents the power output at node $i$ along the edge $(i,j)$, {$z_{i,j}^*$ denotes the \textit{complex conjugate} of $z_{i,j}$}, and $v_j\triangleq |V_j|^2$ and $\ell_e\triangleq|I_e|^2$ define the voltage and current magnitude squares at node $j$ and link $e$, respectively. Note that in BFM phase angles for the voltages and currents, $\arg(V_j)$ and $\arg(I_e)$, are eliminated from the formulation. However, as proved in \citep{FL13a}, this relaxation is exact for radial networks. That is, one can (in polynomial time) uniquely recover the phase angles once a solution to the relaxation is obtained. Finally, each user demand $k\in \cN$ is assigned a control variable $x_k$; if $k \in \cI$, then $x_k\in\{0,1\}$, otherwise,  $x_k \in [0,1]$ for $k \in \cF$. Define  vectors $S\triangleq(S_e)_{e\in\cE}, \ell \triangleq (\ell_e)_{e\in\cE}, x\triangleq(x_k)_{k\in \cN}, v = (v_i)_{i\in \cV^+}$. 

\paragraph{\bfseries The objective:} OPF seeks to assign values to the control vector $x$, complex power vector $S$ as well as current and voltage magnitude vectors $\ell$ and $v$, such that the following {\em concave} non-negative objective function\footnote{Traditionally, the objective is to minimize the generation cost $c(S_{0,1}^{\rm R})$, which is typically a non-decreasing convex function of the active generation power $S_{0,1}^{\rm R}$. In the discrete demand case under study, we combine the minimization of the generation cost with the utility maximization of the satisfied demands by using the function $f_{\textsc{OPF}}(s_0,x)$, where  $f_0(s_0^{\rm R})\triangleq Y-c(S_{0,1}^{\rm R})=Y-c(-s_{0}^{\rm R}))$, for a sufficiently large number $Y$, is a nonnegative concave function, non-decreasing in $s_{0}^{\rm R}).$} 
$$f_{\textsc{OPF}}(s_0, x) = f_0(s_0^{\rm R}) + f_1\big((s_k^{\rm R} x_k)_{k\in\cF} \big) + \sum_{k \in \cI}u_k x_k,$$
is  maximized, without violating the physical and operating constraints described below.


\paragraph{\bfseries The constraints:} Let $\underline v_j, \overline v_j \in\RR^+$ be respectively the minimum and maximum allowable voltage magnitude squares at  node $j$, and $\overline S_{e}, \overline \ell_{e}\in \RR^+$ be the maximum allowable apparent power and current magnitude on edge $e \in \cE$, respectively. 
As customary, it is assumed that the generator voltage $v_0\in \RR^+$ is given. In the above formulation, Eqn.~\raf{p1:con1} is immediate from the definition of the magnitude of the complex power $S_{i,j}=V_iI_{i,j}^*$.  Eqn.~\raf{p1:con2} (in complex variables) captures the power flow conservation rule at node $j$ (see Figure~\ref{f3}). The rule equates the power output at node $i$ along the edge $(i,j)$ minus the power lost on that line ($z_{i,j}\ell_{i,j}=z_{i,j}|I_{i,j}|^2$) to the total power consumed by the loads at node $j$ (namely, $\sum_{k \in \cU_j} s_k x_k$) plus the total power output on the lines outgoing from $j$ (which is $\sum_{t:(j,t)\in \cE} S_{j,t}$).
Eqn.~\raf{p1:con3} is the special case of Eqn.~\raf{p1:con2} applied to node $0$ (assuming an artificial edge $(0,0)$), where the demand $s_0$ is negated to indicate power generation (rather than consumption).  Eqn.~\raf{p1:con4} is a consequence of Ohm's law: $V_i-V_j=z_{i,j}I_{i,j}$, and the definition of power $S_{i,j}=V_iI_{i,j}^*.$ The inequalities in~\raf{p1:con5} and \raf{p1:con7} limit the voltage and current magnitudes at each node and on each line, respectively, to the allowable range. While those in~\raf{p1:con6} cap the apparent power on each link in both directions by the capacity of the link: $|S_{i,j}|\le \overline S_{i,j}$ and $|S_{j,i}|\le \overline S_{i,j}$, where $S_{j,i}=V_jI_{j,i}^*=-V_jI_{i,j}^*=-(V_i-z_{i,j}I_{i,j})I_{i,j}^*=-S_{i,j}+z_{i,j}|I_{i,j}|^2$.

\subsubsection{Assumptions} In tackling OPF, we shall rely on the following practical assumptions.
\begin{itemize}
	\setlength{\itemindent}{.1in}
	\item [{\sf A0:}] $f_0(\cdot)$ is non-decreasing in $s^{\rm R}_0$. {Recall that by definition $f_0(s_0^{\rm R}) = Y-c(-s_{0}^{\rm R}))$, where $c(-s_{0}^{\rm R}) = c(S_{0,1}^{\rm R})$ captures the active power generation cost. As is customary in power systems literature \citep{huang2017sufficient, FL13a, gan2015exact, 6290429}, we treat the generation cost $c(\cdot)$ as a non-decreasing convex function of $S_{0,1}^{\rm R}$. Consequently, one can set $Y$ to be a sufficiently large number such that $f_0(\cdot)$ is non-negative and non-decreasing in $s_{0}^{\rm R}).$}
	
	\item [{\sf A1:}]  $z_e\ge 0$ for all $e\in\cE$, which naturally holds in distribution networks.

	\item[{\sf A2:}]		 $  \underline v_j\le v_0 \le \overline v_j$ for all $j\in\cV^+$. Typically in a distribution network, $v_0$ = 1 (per unit), $\underline v_{j}=(0.95)^2$ and $\overline v_{j}=(1.05)^2$; in other words, a $5\%$ deviation from the nominal voltage is allowed. 
	\item [{\sf A3:}]	$\re(z_{e}^* s_k)\ge 0 \text{ for all } k\in \cN,~e\in \cE.$
	Equivalently, the angle difference between $z_e$ and $s_k$ is at most $ \tfrac{\pi}{2}$.
	
	\item [{\sf A4:}]	$\big|\arg(s_k) - \arg(s_{k'})\big| \le  \tfrac{\pi}{2}$ for any $k, k' \in \cN$. In  practical settings, the so-called {\it load power factor} usually varies between $0.8$ to $1$ \citep{1339347} and thus the maximum phase angle difference between any pair of demands is restricted to be in the range of $[0, 36^{\circ}]$. 
	We also assume $s^{\rm R}_k\ge 0$ for all $k\in \cN$, which always holds in power systems (assuming no power generation at non-root nodes in $\cV^+$).
	
	\item [{\sf A5:}] The range of impedances and demands is quasi-polynomial, that is,
	{\small$$
	 \max\left\{\frac{\max_{e\in\cE}z_{e}^{\rm R}}{\min_{e:z_e^{\rm R}>0}z_{e}^{\rm R}},\frac{\max_{e\in\cE} z_{e}^{\rm I}}{\min_{e:z_e^{\rm I}>0}z_{e}^{\rm I}},\frac{\max_{k\in\cN}s_{k}^{\rm R}}{\min_{k:s_k^{\rm R}>0}s_{k}^{\rm R}},\frac{\max_{k\in\cN}s_{k}^{\rm I}}{\min_{k:s_k^{\rm I}>0}s_{k}^{\rm I}}\right\}=2^{\polylog(m,\tilde{ n})}
	.
	$$}
\end{itemize}
Assumptions {\sf A3}  and  {\sf A4} are motivated, from a theoretical point of view, by the inapproximability results in \citep{khonji2016optimal} (if either one is invalid, the problem cannot be approximated within any polynomial factor unless P=NP). Assumption {\sf A3} holds in reasonable practical settings~\citep{huang2017sufficient}. As clarified in the next subsection, by performing an axis rotation, {\sf A4} implies $s_k\ge 0$. Clearly, under this and {\sf A1}, the reverse power constraint in~\raf{p1:con6} is implied by the forward power one ($|S_e|\le\overline S_e$). Similarly, under {\sf A1},   {\sf A2}  and  {\sf A3}, the voltage upper bounds in \raf{p1:con5} can be dropped, as elaborated in subsection~\ref{dolya}. Lastly, {\sf A5} is required merely for the analysis of the featured approximations and may possibly be bypassed with techniques from~\citep{BGKMW15}.

\subsubsection{Rotational Invariance of {\sc OPF}}\label{dolya}
In the below lemma, it is argued that complex quantities in the OPF formulation (namely, $z_e, s_k$) can be rotated by a fixed angle without affecting the problem's structure. This property allows to replace {\sf A0} and {\sf A4} by the ones listed below.
\begin{itemize}
	\setlength{\itemindent}{.13in}
	\item [{\sf A0$'$}:] $f_0(s_0^{\rm R} \cos \phi + s_0^{\rm I} \sin \phi)$ is  non-decreasing in $s_0^{\rm R}, s_0^{\rm I}$.
	\item [{\sf A4$'$:}] $s_k \ge 0$ for all $k\in \cN$.
\end{itemize}
Note that  {\sf A1}  and {\sf A4$'$} already imply {\sf A3}.
\begin{lemma}\label{lem:rot}
	Assume {\sf A4} and suppose that $s_k$, for all $k\in\cN$,  and $z_e$, for all $e\in\cE$, are rotated by an angle $\phi \triangleq \min\{\max_{k\in \cN} -\arg(s_k),0\}\in[0,\frac{\pi}{2}]$. Denote the resulting {\sc OPF} problem  by {\sc OPF$^\phi$}:
	\begin{align}
	\textsc{(OPF$^\phi$)} \quad &\max_{\substack{s_0, x, v,\ell, S \;\;}} f_{\textsc{OPF}}(s_0e^{-{\bf i}\phi}, x),  \notag \\
	\text{s.t.} \ &\raf{p1:con1}-\raf{p1:con9}, \text{ with $z_e$ replaced by $z_ee^{{\bf i}\phi}$, and $s_k$ replaced by $s_ke^{{\bf i}\phi}$ }	. \notag
	\end{align}
	Then {\sc OPF$^\phi$} is equivalent to {\sc OPF} and satisfies assumptions {\sf A0$'$}, {\sf A1}, {\sf A2}, {\sf A3} and {\sf A4$'$}.
	
\end{lemma}
\begin{proof}
	One can easily show that a feasible solution $F=(s_0,x,v,\ell,S)$ to {\sc (OPF$^\phi$)} can be converted to a feasible solution $\bar{\bar F}=(\bar{\bar s}_0, x,v,\ell,\bar{\bar S})$ to {\sc OPF}, such that $\bar{\bar S}_{i,j}\triangleq S_{i,j} e^{-{\bf i}\phi}, \bar{\bar s}_0\triangleq s_0 e^{-{\bf i} \phi}$ are rotated by $\phi$, and vise versa. 
	Moreover, the two objective functions are equal. It is  immediate to see that assumptions {\sf A0$'$}  {\sf A1}, {\sf A2}, {\sf A3}, and  {\sf A4$'$} hold for {\sc OPF$^\phi$}.  \hfill$\blacksquare$
\end{proof}

Hereafter, we implicitly consider the rotated problem which, with a slight abuse of notation, is simply denoted by OPF.

\subsubsection{Exact Second Order Cone Relaxation}\label{sec:rlx}
As observed from the preceding formulation, OPF's feasible set is non-convex due to the quadratic equality constraint~\raf{p1:con1}. Replacing this by $\ell_{i,j} \ge \frac{|S_{i,j}|^2}{v_i}$, one obtains an SOCP relaxation  of \textsc{OPF}\footnote{Note that Cons.~\raf{p2:con1} can be rewritten as
\begin{align*}
\Bigg\|\left(\begin{array}{c}
2S_{i,j}^{\rm R}\\
2S_{i,j}^{\rm I}\\
\ell_{i,j}-v_i
\end{array}
\right)\Bigg\|_2 \le \ell_{i,j}+v_i\,.
\end{align*}}, defined below and denoted by \textsc{cOPF}.
\begin{align}
\textsc{(cOPF)} &\max_{\substack{s_0, x, v,\ell, S \;\;}} f_{\textsc{OPF}}(s_0, x)  \notag \\
\text{s.t.} \ &\raf{p1:con2}-\raf{p1:con9},\notag \\
&	\ell_{i,j} \ge  \frac{|S_{i,j}|^2}{v_i},~\forall (i,j) \in \cE. \label{p2:con1} 
\end{align}
Let {\sc rcOPF} be the relaxation of {\sc cOPF} where the integrality constraints in~\raf{p1:con8} are replaced by $x_k \in [0,1]$ for all $k \in \cN$. For a given $\hat x\in[0,1]^{\tilde{ n}}$, define by  {\sc cOPF}$[\hat x]$ the restriction of {\sc cOPF} where $x=\hat x$. 

Recently, studies in~\citep{low2014convex2,huang2017sufficient,gan2015exact} presented sufficient conditions for \textsc{cOPF} to have an optimal solution in which Cons.~\raf{p2:con1} holds with equality. For current purposes, we avail of the following lemma which is a slightly simplified version of that in \citep{huang2017sufficient} and is proved in Section~\ref{lem44}. 

\begin{lemma}\label{lem:exact} Under assumptions~{\sf A0}, {\sf A1}, {\sf A2}, and {\sf A3}, for any given $x'\in[0,1]^{\tilde{ n}}$, there exists an optimal solution  $F'=(s_0', x', v',\ell',S')$ of {\sc cOPF}$[x']$ that satisfies  $\ell_{i,j} = \frac{|S_{i,j}'|^2}{v_i'}$ for all  $(i,j)\in \cE.$ Such a solution can be found in polynomial time. 
\end{lemma}

\subsection{Reduction Scheme}\label{sec:pre}
 
Having defined OPF formally, we next present the developed technique that obtains approximations for OPF on path distribution networks from LP-based approximations intended for separable $d$-USFP.
\begin{lemma}\label{feas}
	Let $F'=\big(s'_0,x',  v', \ell',  S'\big)$ be a feasible solution for \textsc{rcOPF}. Let $\bar x\in[0,1]^{\tilde{ n}}$ be such that 
	\begin{align}
	\sum_{k\in\cI}u_k\bar x_k&\ge \sum_{k\in\cI}u_k x_k'-\eps f_{\textsc{OPF}}(s_0',x'), \text{ for some }\varepsilon\in[0,1] \label{feas:con0}\\
	\sum_{k \in \cN}  \re\Big( \sum_{(h,t)\in \cP_k \cap \cP_j} z^*_{h,t} s_k\Big) \bar x_k  &\le \sum_{k \in \cN}  \re\Big( \sum_{(h,t)\in \cP_k \cap \cP_j} z^*_{h,t} s_k\Big) x'_k\quad \forall (i,j) \in \cE , \label{feas:con1}\\
	\sum_{k \in \cN_j} s_k^{\rm R} \bar x_k  &\le \sum_{k \in \cN_j} s_k^{\rm R} x'_k\quad \forall (i,j)\in \cE, \label{feas:con2}\\
	\sum_{k \in \cN_j} s_k^{\rm I} \bar x_k &\le\sum_{k \in \cN_j} s_k^{\rm I}  x_k'\quad \forall (i,j)\in \cE,\label{feas:con3}\\
	\bar x_k&= x_k'\quad \forall k\in\cF,\label{feas:con4}
	\end{align} 
	where $f_{\textsc{OPF}}(\cdot)$ is the objective function of OPF. Then, under assumptions~{\sf A0$'$}, {\sf A1}, {\sf A2}, {\sf A3} and {\sf A4$'$}, \textsc{rcOPF$[\bar x]$} has a feasible solution $\tilde F=\big(\tilde s_0,\tilde x,  \tilde v, \tilde \ell, \tilde  S\big)$ such that $f_{\textsc{OPF}}(\tilde s_0,\tilde x)\ge (1-\eps) f_{\textsc{OPF}}(s_0',x')$, where \textsc{rcOPF$[\bar x]$} denotes the restriction of {\sc rcOPF} with $x$ set to $\bar x$. 
\end{lemma}

Observe that, in Lemma~\ref{feas} (which is proved in Section~\ref{lprf}), the inequalities \raf{feas:con1}, \raf{feas:con2} and \raf{feas:con3} taken together form a single-source separable $d$-USFP with $d=3$. Indeed, for $k\in\cI$ and $e=(i,j)\in\cE$, define
\begin{align*}
f_k^1(e)=\re\Big( \sum_{e'\in \cP_k \cap \cP_j} z^*_{e'} s_k\Big), \quad
f_k^2(e)=\left\{\begin{array}{ll}
s_k^{\rm R}&\text{ if }k\in\cN_j \\
0&\text{otherwise,}
\end{array}
\right.\quad
f_k^3(e)&=\left\{\begin{array}{ll}
s_k^{\rm I}&\text{ if }k\in\cN_j \\
0&\text{otherwise.}
\end{array}
\right.
\end{align*}
Note that $f^1_k$ is monotone non-decreasing on $\cE$ when ordered by distance from the root, while  $f^2_k$ and $f^3_k$ are monotone non-decreasing considering the reverse order on $\cE$. Moreover, these functions are of the form~\raf{funct} (i.e., separability condition in $d$-USFP). For $r=2$ (similarly, for $r=3$), set $T_2=1$, $a_k^{2,1}:=s_k^{\rm R}$, $e_k^2=\hat e_k^2:= e_{j(k)}$, $b^{2,1}(e):=1$.
As for $r=1$, note that
\begin{align}
\label{fff0}
f_k^1(e)=\Big(\sum_{e'\in \cP_k \cap \cP_j} z^{\rm R}_{e'}\Big) s^{\rm R}_k+\Big(\sum_{e'\in \cP_k \cap \cP_j} z^{\rm I}_{e'}\Big) s^{\rm I}_k \text{ for } \forall e=(i,j)\in\cE .
\end{align}
Thus, setting $T_1=2$, $a_k^{1,1}:= \re( s_k)$, $a_k^{1,2}:=\im( s_k)$, $e^1_k:=e_1$, $\hat e^1_k:=e_{j(k)}$, $\tilde{b}^{1,1}((i,j)):=\sum_{e'\in \cP_j} \re( z_{e'})$ and $\tilde{b}^{1,2}((i,j)):=\sum_{e'\in \cP_j} \im( z_{e'})$ writes $f_k^1(e)$ in the form \raf{funct}.

The above arguments coupled with Lemma~\ref{feas}, imply the following theorem.

\begin{theorem} Under assumptions~{\sf A0$'$}, {\sf A1}, {\sf A2}, {\sf A3}, {\sf A4$'$}, and {\sf A5}, there is a quasi-polynomial time algorithm that for any $\epsilon \in (0,1)$ produces a $(1- \eps)$-approximation for \textsc{OPF} on line networks with single substation generator. \label{qptasopf}
\end{theorem}
\begin{proof}
	Let $\widehat{\OPT}$ be the optimal objective value of OPF. Consider the approximation scheme detailed in Alg.~\ref{alg:qptas-}, which is the analog of Alg.~\ref{alg:qptas} for OPF.
	
	\begin{algorithm}[!htb]
		\caption{ {\sc QPTAS-OPF} \label{alg:qptas-} }
		\begin{algorithmic}[1]
			\Require An approximation parameter $\epsilon\in(0,1)$; {\sc OPF} input $v_0;(\underline v_j, \overline v_j)_{j\in \cV^+}; (\overline S_e, \overline  \ell_e, z_e)_{e\in \cE}$
			\Ensure	A solution $\hat F$ to {\sc OPF} such that $f_{\textsc{OPF}}(\hat F) \ge (1-O(\epsilon)) \widehat{\OPT}$
			
			\For {each selection {$\Big(\cL = (\cL^q)_{q\in\cQ}, h = \big(h^q=(h^{q,p,r})_{p\in[P_r],~r\in[d]}\big)_{q\in\cQ}\Big)$ such that}  $\cL^q\subseteq \cI$, $|\cL^q| \le \frac{\sum_{r=1}^dP_r}{\epsilon^2}$  and $h^{q,p,r}\in F^r$ } \label{alg:guess-} 
			\If{{\sc rcOPF$[\cL,h]$} is feasible} \label{qptas:mc-feas-}
			
			\State $F' \leftarrow$  Solution of {\sc rcOPF$[\cL,h]$} \label{qptas:cvx}
			\For{$q\in\cQ$}\label{qptas:Q-}
			\State Let $\cS^q$ be given by~\raf{small}
			\For {every $(h,\epsilon)$-restricted profile $g^q$}\label{qptas:prof-}
			\State$ (\hat x_k)_{k\in \cS^q} \leftarrow $ Integral vector returned by applying Lemma~\ref{l4}  with vector $h^q$, and {$(\tilde x_k)_{k\in \cS^q}=(x'_k)_{k\in \cS^q}$} \label{qptas:mc-round-}
			\EndFor
			\EndFor
			\State$\bar x_k\leftarrow \left\{\begin{array}{ll}
			\hat x_k &\text{ if }k\in \bigcup_{q\in\cQ}\cS^q, \\
			x'_k&\text { if }{k\in \cN\setminus(\bigcup_{q\in\cQ}\cS^q)} 
			\end{array}\right.$  \label{qptas:mc-round--}
			\State $\tilde F \leftarrow$  Solution of \textsc{cOPF$[\bar x]$}  \label{qptas:r3}
			\If{$f_{\textsc{OPF}}(\tilde F) > f_{\textsc{OPF}}(\hat F')$ }
			\State $ \hat F' \leftarrow \tilde F$
			\EndIf 
			\EndIf
			
			\EndFor
			\State Apply Lemma~\ref{lem:exact} to convert $\hat F'$ to a feasible solution $\hat F$ for OPF
			\State \Return $\hat F$
		\end{algorithmic}
	\end{algorithm}
	
	Similar to Alg.~\ref{alg:qptas}, the algorithm guesses the set of large demands $\cL^q\subseteq\cI^q$ in the optimal solution for each group $q\in\cQ$, and the peaks $h^{q,p,r}$, within $1+\epsilon$, of the small demands in the optimal solution within the interval $\cE_p^r$. Let $\cL=(\cL^q)_{q\in\cQ}$ and $h^q=(h^{q,p,r})_{p\in[P_r],~r\in[d]}$ where $h^{q,p,r}\in F^r$. Define a restrictive version of {\sc rcOPF}, denoted by {\sc rcOPF$[\cL,h]$}, which enforces that $x_{k} = 1$ for all ${k} \in \cL^q$ and $q\in\cQ$ and that the peak total contribution of the small demands in group $q$ within the interval $\cE_p^r$ is at most $h^{q,p,r}$: $\sum_{k\in\cS^q}\overline f_k^{p,r} x_k\le(1+\epsilon) h^{q,p,r}$.
	
	\begin{align}
	\textsc{(rcOPF}[\cL,h]\textsc{)}\quad &\max_{\substack{s_0, x,v, \ell, S}} f_{\textsc{OPF}}(s_0,x),  \notag \\
	\text{s.t.} \ & \raf{p1:con2}-\raf{p1:con7},\raf{p1:con9},\raf{p2:con1}\\
	&\sum_{k\in\cS^q}\overline{f}_k^{p,r}x_k\le h^{q,p,r},~\forall p\in[P_r],~\forall r\in[d],~\forall  q\in\cQ \label{e111}\\
	& x_k = 0,\quad \forall k \in \cI\setminus \bigcup_{q\in\cQ}(\cL^q\cup\cS^q)\label{e112-}\\
	& x_k = 1,\quad \forall k \in \cL^q,~\forall q\in\cQ\label{e112}\\
	& x_k \in [0,1],\quad \forall k \in \cF\cup (\bigcup_{q\in\cQ}\cS^q). \label{e113}
	\end{align}
	 Here, the set of small demands within group $q\in\cQ$ is
	\begin{align}\label{small-}
	\cS^q=\left\{k\in\cI^q:~\underline{f}_k^{p,r}\le B^{q,p,r}\text{ for all } p\in[P_r],~r\in[d]\right\},
	\end{align}
	where $B^{q,p,r}=\epsilon^2\left[h^{q,p,r}+\sum_{k\in\cL^q}\overline{f}_k^{p,r}\right].$
	Given a feasible solution $ F'=\big(s'_0,x', $ $ v', \ell',  S'\big)$ to {\sc rcOPF$[\cL,h]$}, Alg.~\ref{alg:qptas-} applies Lemma~\ref{l4} with $\tilde{x}=x'$. By the lemma, one can find (in polynomial time) an integral solution $\hat x$ satisfying conditions (i) and (ii). Next, the algorithm recalculates $s_0, S, \ell, v$ utilizing the program~\textsc{cOPF$[\bar x]$} given in Section~\ref{sec:rlx}, and then applies Lemma~\ref{lem:exact} to obtain a feasible solution to OPF.
	
	Define
	{\footnotesize
		\begin{align}
		\label{qp-bd}
		\tilde{M}:=&\max\left\{\frac{\overline z}{\underline z},\max_{r}\frac{\overline f^r}{\underline f^r}\right\} \nonumber\\
		&=\max\left\{\frac{\overline z}{\underline z},\max_{k,k'\in\cI}\frac{s_k^{\rm R}}{ s_{k'}^{\rm R}},\max_{k,k'\in\cI}\frac{s_k^{\rm I}}{ s_{k'}^{\rm I}},\max_{k,k'\in\cI,~(i,j),(i',j')\in\cE}\frac{\re\Big( \sum_{e'\in \cP_k \cap \cP_j} z^*_{e'} s_k\Big)}{\re\Big( \sum_{e'\in \cP_{k'} \cap \cP_{j'}} z^*_{e'} s_{k'}\Big)}\right\},
		\end{align}
	}
	where $\underline z:=\min\{\min_{e:\re(z_e)>0}\re(z_{e}),\min_{e:\im(z_e)>0}\im(z_{e})\}$ and $\overline z:=\max_{e\in\cE}$  $\max\{\re(z_{e}),\im(z_{e})\}$.
	
	In what follows, we prove that, for any fixed $\eps\in(0,1)$, Alg.~\ref{alg:qptas-} arrives at  a $(1- \eps)$-approximation in time $(\frac{\tilde{n}\log (\tilde{ n}m\tilde{ M})}{\eps})^{O(\log^9( \frac{\tilde{n}m\tilde{M}}{\eps})/\eps^2)}$. 
	
	Let $\epsilon:=\frac{\eps}{3(2\beta+1)}$, where $\beta=\max_{r\in\cH^q}2\left(2C_r+\alpha P_r\right)=O(\log^2 (m\tilde{M}))$. 
	The number of possible choices for each $\cL^q$ in step~\ref{alg:guess-} of Alg.~\ref{alg:qptas-} is at most $\tilde{n}^{\sum_{r=1}^dP_r/\epsilon^2}$, where $\tilde{ n}=|\cN|$. Thus, with $d=3$, $P_1=O(\log (m\tilde{M}))$, $P_2=P_3=1$, $T_1=2$, $T_2=T_3=1$, $Q\leq \overline Q^{\sum_{r=1}^dT_r}$, and $\overline Q=O(\log\frac{\tilde{n}\tilde{M}}{\epsilon})$, hence the number of possible choices for $\cL$ is at most \begin{align}\label{ch-L}
	\tilde{n}^{\sum_{r=1}^dP_rQ/\epsilon^2}\le \tilde{n}^{\sum_{r=1}^dP_r\overline Q^{{\sum_{r=1}^dT_r}}/\epsilon^2}= \tilde{n}^{O(\log(m\tilde{M})\log^4( \frac{\tilde{ n}\tilde{ M}}{\epsilon})/\epsilon^2)}.
	\end{align} The number of choices for each $h^q=(h^{q,p,r})_{p\in[P_r],~r\in[d]}$ is $$\overline F^{\sum_{r=1}^dP_r}=O\Bigg(\Big(\frac{\log(\tilde{n}\tilde{M})}{\epsilon}\Big)^{\log (m\tilde{M})}\Bigg)\,,$$ {and that of for} $Q$ in step~\ref{qptas:Q-} is 
	\begin{align}
	\label{ch-Q}
	\overline{Q}^{\sum_{r=1}^d {T}_r}\le\log^4\left( \frac{\tilde{ n}\tilde{ M}}{\epsilon}\right),
	\end{align}
	giving at most
	\begin{align}
	\label{ch-h}
	O\left(\left(\left(\frac{\log(\tilde{n}\tilde{ M})}{\epsilon}\right)^{\log (m\tilde{ M})}\right)^{ Q}\right)&=O\left(\left(\left(\frac{\log(\tilde{ n}\tilde{ M})}{\epsilon}\right)^{\log (m\tilde{ M})}\right)^{\overline Q^{\sum_{r=1}^d {T}_r}}\right)\nonumber\\
	&=O\left(\left(\frac{\log(\tilde{ n}\tilde{ M})}{\epsilon}\right)^{\log (m\tilde{ M})\log^4( \frac{\tilde{ n}\tilde{ M}}{\epsilon})}\right)
	\end{align} choices for $h=(h^q)_{q\in\cQ}$ in step~\ref{alg:guess-}.
	The number of choices for the $\epsilon$-restricted profiles in step~\ref{qptas:prof-} is bounded from above by
	$
	m^{\sum_{r=1}^dP_r/\epsilon}=m^{O(\log{(m\tilde{ M})}/\epsilon)}.
	$
	Thus, the bound on the running time follows from this and~\raf{ch-L},\raf{ch-h},\raf{ch-Q}. 
	
	We now argue that the solution $\hat F$ outputted by Alg.~\ref{alg:qptas-} is $(1-O(\epsilon))$-approximation for {\sc OPF}. 
	Let $F^*=\big(s^*_0,x^*,  v^*, \ell^*,  S^*\big)$ be an optimal solution for {\sc OPF} of objective value $\widehat{\OPT} = f_{\textsc{OPF}}(F^*)$. By the definition of $\hat\cI$, we have
	\begin{align}\label{en0}
	\sum_{k\in\cI\setminus\hat\cI}u_k\le\epsilon\widehat{\OPT}\le\epsilon f_{\textsc{OPF}}(F^*).	
	\end{align}
	Define  $\cT^*\triangleq\{k \in \hat\cI \mid x^*_k = 1\}$ and $(h^*)^{q,p,r}=\sum_{k\in\cT^*\cap\cI^q}\overline f_k^{p,r}$, for $p\in[P_r],$ $r\in[d]$ and $q\in\cQ$.  
	Let $(\cL^*)^q:=\{k\in \cI^q\cap\cT^*:~\underline{f}_k^{p,r}>\epsilon^2(h^*)^{q,p,r}\text{  for some } p\in[P_r],~\text{and some } r\in[d]\}$ be the set of large demands within group $\cI^{q}$ in the optimal solution, and let $(\cS^*)^q:=\cI^q\cap\cT^*\setminus(\cL^*)^q$ be the set of ``small" demands within the same group. Note by this definition that $|(\cL^*)^q|\le\frac{\sum_{r=1}^dP_r}{\epsilon^2}$, and thus $\cL^*=((\cL^*)^q)_{q\in\cQ}$ and $h=(h^q)_{q\in\cQ}$ will be one of the guesses considered by the algorithm in step~\ref{alg:guess-}. Let us focus on this particular iteration of the loop in step~\ref{alg:guess-}. Let $h^{q,p,r}=(1+\epsilon)^{\ell'}\underline{f}^{r}$, where $\ell'$ is the smallest integer (including $-\infty$) such that  $h^{q,p,r}+\sum_{k\in\cL^q}\overline f_k^{p,r}\ge  (h^*)^{q,p,r}$. Note that $h^{q,p,r}\in F^r$,  and
	\begin{align}\label{rel}
	\frac{1}{1+\epsilon}h^{q,p,r}+\sum_{k\in(\cL^*)^q}\overline f_k^{p,r}\le (h^*)^{q,p,r}\le h^{q,p,r}+\sum_{k\in(\cL^*)^q}\overline f_k^{p,r}.
	\end{align}
	
	Moreover, for any $k\in(\cS^*)^q$, $q\in\cQ$, $p\in[P_r]$, and $r\in[d]$, we have by \raf{rel},
	$$
	\underline f_k^{p,r}\le\epsilon^2 (h^*)^{q,p,r}\leq\epsilon^2\left(h^{q,p,r}+\sum_{k\in(\cL^*)^q}\overline f_k^{p,r}\right),
	$$
	and hence $(\cS^*)^q\subseteq\cS^q$. Note also that 
	\begin{align}
	\label{bbb2}
	B^{q,p,r}&=\epsilon^2\left[h^{q,p,r}+\sum_{k\in(\cL^*)^q}\overline{f}_k^{p,r}\right]\nonumber\\
	&\le\epsilon^2\left[h^{q,p,r}+(1+\epsilon)\sum_{k\in(\cL^*)^q}\overline{f}_k^{p,r}\right]\le \epsilon^2(1+\epsilon)(h^*)^{q,p,r}.
	\end{align}
	Furthermore, $x^*$ is feasible for the constraint~\raf{e111} as 
	$$
	\sum_{k\in\cS^q}\overline f_k^{p,r}x_k^*=\sum_{k\in(\cS^*)^q}\overline f_k^{p,r}x_k^*=\sum_{k\in(\cS^*)^q}\overline f_k^{p,r}=(h^*)^{q,p,r}-\sum_{k\in(\cL^*)^q}\overline f_k^{p,r}\le h^{q,p,r}. 
	$$
	It follows that $F^*$ is feasible for~{\sc R1$[\cL,h]$}, implying by~\raf{en0} that  the solution $F'$ obtained in step~\ref{qptas:cvx} of the algorithm satisfies  
	\begin{equation}
	f_{\textsc{OPF}}(F') \ge  (1-\epsilon)f_{\textsc{OPF}}(F^*). \label{eq:qptas-cvx}
	\end{equation}
	For each $q\in\cQ$, there is an $(h,\epsilon)$-restricted profile $g^q$ and an integral solution $ (\hat x_k)_{k\in \cS^q}  $ that satisfy Lemma~\ref{l4}. Since all the possible $(h,\epsilon)$-restricted profiles are probed, the profile $g^q$ will be found in one of the iterations in the loop in line~\ref{qptas:prof-}. Let us consider this iteration. By condition (i) of the lemma, $\sum_{k\in\cS^q}f_k^r(e)\hat x_k\leq \sum_{k\in\cS^q}f_k^r(e){x}'_k$ for all $e\in\cE$ and $r\in[d]$, which implies that conditions~\raf{feas:con1}-\raf{feas:con4} of Lemma~\ref{feas} hold for the vector $\bar x$, defined in line~\ref{qptas:mc-round--} of Alg.~\ref{alg:qptas-}.
	
	At this point, following exactly the same lines as in the proof of Theorem~\ref{th:mcqptas-}, it can be shown that
	
	\begin{align*}
	\sum_{k\in\cI}u_k\bar x_k \ge \sum_{k \in \hat\cI}u_kx_k'-3\epsilon(2\beta+1)f_{\textsc{OPF}}(F').
	\end{align*}
	Thus condition~\raf{feas:con0} in Lemma~\ref{feas} is satisfied with $\eps=3\epsilon(2\beta+1)$ implying that $\tilde F$ is a feasible solution for \textsc{cOPF}, and hence for OPF by Lemma~\ref{lem:exact},  with $f_{\textsc{OPF}}(\tilde F)\ge (1-\eps) f_{\textsc{OPF}}(F')\ge(1-\eps)f_{\textsc{OPF}}(F^*)$. \hfill$\blacksquare$
\end{proof}

\vspace*{5pt}
\begin{remark} Following arguments analogous to those in the above proof, it is conceivable to generalize the logarithmic approximation devised in Section~\ref{sec:logarithmic} to OPF on line networks with single substation generator, provided assumptions~{\sf A0$'$}, {\sf A1}, {\sf A2}, {\sf A3}, {\sf A4$'$} and {\sf A5} hold. To this end, however, an additional constant factor would be lost in the approximation ratio for bounding the capacities (i.e., the right hand sides of inequalities~\raf{feas:con1}, \raf{feas:con2} and \raf{feas:con3}).
\end{remark}

\section{Concluding Remarks}\label{conclrem}

This study defined a novel generalization of UFP, dubbed as $d$-USFP, and bridged it with AC OPF, which is a fundamental problem in power systems engineering. In a preliminary step towards tackling this extended problem, we devised a QPTAS and an efficient logarithmic approximation for its single-source variant. Leveraging the connection between separable $d$-USFP and AC OPF, a (kind of) black-box reduction is developed that, under some mild conditions, allows one to convert an approximation for the former problem to that of for AC OPF on line distribution networks with discrete demands. It's noteworthy that this reduction applies only to algorithms that depend on LP-rounding techniques, hence the focus of the present study on LP-based approximations. Whereas for future work, it would be interesting to generalize and extend the known alternative techniques (e.g., the surveyed combinatorial and dynamic programming based ones) to $d$-USFP, consequently improving upon the current results. As from power systems perspective, one future avenue to explore, would be extension of the established framework to a more practical setting with multiple generation sources and tree networks.

\bibliographystyle{spbasic}      
\bibliography{paper}

\begin{thebibliography}{44}
\providecommand{\natexlab}[1]{#1}
\providecommand{\url}[1]{{#1}}
\providecommand{\urlprefix}{URL }
\expandafter\ifx\csname urlstyle\endcsname\relax
  \providecommand{\doi}[1]{DOI~\discretionary{}{}{}#1}\else
  \providecommand{\doi}{DOI~\discretionary{}{}{}\begingroup
  \urlstyle{rm}\Url}\fi
\providecommand{\eprint}[2][]{\url{#2}}

\bibitem[{Adamaszek et~al.(2016)Adamaszek, Chalermsook, Ene, and
  Wiese}]{Adamaszek2018}
Adamaszek A, Chalermsook P, Ene A, Wiese A (2016) Submodular unsplittable flow
  on trees. In: International Conference on Integer Programming and
  Combinatorial Optimization, Springer, pp 337--349

\bibitem[{Albers et~al.(1999)Albers, Arora, and Khanna}]{Albers:1999}
Albers S, Arora S, Khanna S (1999) Page replacement for general caching
  problems. In: SODA, Citeseer, vol~99, pp 31--40

\bibitem[{Anagnostopoulos et~al.(2014)Anagnostopoulos, Grandoni, Leonardi, and
  Wiese}]{anagnostopoulos2014mazing}
Anagnostopoulos A, Grandoni F, Leonardi S, Wiese A (2014) A mazing (2+
  $\varepsilon$)-approximation for unsplittable flow on a path. In: Proceedings
  of the twenty-fifth annual ACM-SIAM symposium on Discrete algorithms, Society
  for Industrial and Applied Mathematics, pp 26--41

\bibitem[{Bansal et~al.(2006)Bansal, Chakrabarti, Epstein, and
  Schieber}]{BCES06}
Bansal N, Chakrabarti A, Epstein A, Schieber B (2006) A quasi-ptas for
  unsplittable flow on line graphs. In: Proceedings of the thirty-eighth annual
  ACM symposium on Theory of computing, ACM, pp 721--729

\bibitem[{Bansal et~al.(2014)Bansal, Friggstad, Khandekar, and
  Salavatipour}]{BFKS14}
Bansal N, Friggstad Z, Khandekar R, Salavatipour MR (2014) A logarithmic
  approximation for unsplittable flow on line graphs. ACM Transactions on
  Algorithms (TALG) 10(1):1

\bibitem[{Bar-Noy et~al.(2001)Bar-Noy, Bar-Yehuda, Freund, Naor, and
  Schieber}]{Bar-Noy2001}
Bar-Noy A, Bar-Yehuda R, Freund A, Naor J, Schieber B (2001) A unified approach
  to approximating resource allocation and scheduling. Journal of the ACM
  (JACM) 48(5):1069--1090

\bibitem[{Baran and Wu(1989)}]{baran19266}
Baran M, Wu FF (1989) {Optimal sizing of capacitors placed on a radial
  distribution system}. IEEE Transactions on Power Delivery 4(1):735--743,
  \doi{10.1109/61.19266}

\bibitem[{Batra et~al.(2015)Batra, Garg, Kumar, Momke, and Wiese}]{BGKMW15}
Batra J, Garg N, Kumar A, Momke T, Wiese A (2015) New approximation schemes for
  unsplittable flow on a path. In: Proceedings of the twenty-sixth annual
  ACM-SIAM symposium on Discrete algorithms, Society for Industrial and Applied
  Mathematics, pp 47--58

\bibitem[{Bose et~al.(2015)Bose, Gayme, Chandy, and Low}]{BGCL15}
Bose S, Gayme DF, Chandy KM, Low SH (2015) Quadratically constrained quadratic
  programs on acyclic graphs with application to power flow. IEEE Transactions
  on Control of Network Systems 2(3):278--287

\bibitem[{{Briglia} et~al.(2017){Briglia}, {Alaggia}, and {Paganini}}]{7926079}
{Briglia} E, {Alaggia} S, {Paganini} F (2017) Distribution network management
  based on optimal power flow: Integration of discrete decision variables. In:
  2017 51st Annual Conference on Information Sciences and Systems (CISS), pp
  1--6, \doi{10.1109/CISS.2017.7926079}

\bibitem[{Calinescu et~al.(2002)Calinescu, Chakrabarti, Karloff, and
  Rabani}]{3-540-47867}
Calinescu G, Chakrabarti A, Karloff H, Rabani Y (2002) Improved approximation
  algorithms for resource allocation. In: International Conference on Integer
  Programming and Combinatorial Optimization, Springer, pp 401--414

\bibitem[{Carpentier(1962)}]{C62}
Carpentier J (1962) Contribution a l’etude du dispatching economique.
  Bulletin de la Societe Francaise des Electriciens 3(1):431--447

\bibitem[{Chakrabarti et~al.(2007)Chakrabarti, Chekuri, Gupta, and
  Kumar}]{CCGK07}
Chakrabarti A, Chekuri C, Gupta A, Kumar A (2007) Approximation algorithms for
  the unsplittable flow problem. Algorithmica 47(1):53--78

\bibitem[{{Chapman} et~al.(2013){Chapman}, {Verbič}, and {Hill}}]{6629395}
{Chapman} AC, {Verbič} G, {Hill} DJ (2013) A healthy dose of reality for
  game-theoretic approaches to residential demand response. In: 2013 IREP
  Symposium Bulk Power System Dynamics and Control - IX Optimization, Security
  and Control of the Emerging Power Grid, pp 1--13,
  \doi{10.1109/IREP.2013.6629395}

\bibitem[{{Chau} et~al.(2018){Chau}, {Elbassioni}, and {Khonji}}]{8637153}
{Chau} SCK, {Elbassioni} K, {Khonji} M (2018) Combinatorial Optimization of
  Alternating Current Electric Power Systems. \doi{10.1561/3100000017}

\bibitem[{Chekuri et~al.(2007)Chekuri, Mydlarz, and Shepherd}]{CMS07}
Chekuri C, Mydlarz M, Shepherd FB (2007) Multicommodity demand flow in a tree
  and packing integer programs. ACM Transactions on Algorithms (TALG) 3(3):27

\bibitem[{Chrobak et~al.(2012)Chrobak, Woeginger, Makino, and Xu}]{Chrobak2012}
Chrobak M, Woeginger GJ, Makino K, Xu H (2012) Caching is hard—even in the
  fault model. Algorithmica 63(4):781--794

\bibitem[{Cook et~al.(1998)Cook, Faber, Marathe, Srinivasan, and
  Sussmann}]{BFb0055088}
Cook D, Faber V, Marathe M, Srinivasan A, Sussmann YJ (1998) Low-bandwidth
  routing and electrical power networks. In: International Colloquium on
  Automata, Languages, and Programming, Springer, pp 604--615

\bibitem[{Elbassioni et~al.(2019)Elbassioni, Karapetyan, and
  Nguyen}]{Elbassioni2019}
Elbassioni K, Karapetyan A, Nguyen TT (2019) Approximation schemes for
  r-weighted minimization knapsack problems. Annals of Operations Research
  279(1):367--386

\bibitem[{Farivar and Low(2013)}]{FL13a}
Farivar M, Low SH (2013) Branch flow model: Relaxations and
  convexification—part i. IEEE Transactions on Power Systems 28(3):2554--2564

\bibitem[{Frank et~al.(2012)Frank, Steponavice, and Rebennack}]{F12a}
Frank S, Steponavice I, Rebennack S (2012) Optimal power flow: a bibliographic
  survey i. Energy Systems 3(3):221--258

\bibitem[{Gan et~al.(2015)Gan, Li, Topcu, and Low}]{gan2015exact}
Gan L, Li N, Topcu U, Low SH (2015) Exact convex relaxation of optimal power
  flow in radial networks. IEEE Transactions on Automatic Control 60(1):72--87

\bibitem[{Grandoni et~al.(2018)Grandoni, Momke, Wiese, and
  Zhou}]{Grandoni3188745}
Grandoni F, Momke T, Wiese A, Zhou H (2018) A (5/3+
  $\varepsilon$)-approximation for unsplittable flow on a path: placing small
  tasks into boxes. In: Proceedings of the 50th Annual ACM SIGACT Symposium on
  Theory of Computing, ACM, pp 607--619

\bibitem[{Grandoni et~al.(2022{\natexlab{a}})Grandoni, M\"{o}mke, and
  Wiese}]{10.1145/3519935.3519959}
Grandoni F, M\"{o}mke T, Wiese A (2022{\natexlab{a}}) A ptas for unsplittable
  flow on a path. In: Proceedings of the 54th Annual ACM SIGACT Symposium on
  Theory of Computing, Association for Computing Machinery, New York, NY, USA,
  STOC 2022, p 289–302

\bibitem[{Grandoni et~al.(2022{\natexlab{b}})Grandoni, M\"{o}mke, and
  Wiese}]{1.9781611977073.39}
Grandoni F, M\"{o}mke T, Wiese A (2022{\natexlab{b}}) {Unsplittable Flow on a
  Path: The Game!} In: Proceedings of the 2022 Annual ACM-SIAM Symposium on
  Discrete Algorithms (SODA), pp 906--926, \doi{10.1137/1.9781611977073.39}

\bibitem[{Hall and Magazine(1994)}]{Hall1994}
Hall NG, Magazine MJ (1994) Maximizing the value of a space mission. European
  journal of operational research 78(2):224--241

\bibitem[{Hijazi et~al.(2017)Hijazi, Coffrin, and Hentenryck}]{Hijazi2017}
Hijazi H, Coffrin C, Hentenryck PV (2017) Convex quadratic relaxations for
  mixed-integer nonlinear programs in power systems. Mathematical Programming
  Computation 9(3):321--367

\bibitem[{Huang et~al.(2017)Huang, Wu, Wang, and Zhao}]{huang2017sufficient}
Huang S, Wu Q, Wang J, Zhao H (2017) A sufficient condition on convex
  relaxation of ac optimal power flow in distribution networks. IEEE
  Transactions on Power Systems 32(2):1359--1368

\bibitem[{{Karapetyan} et~al.(2018){Karapetyan}, {Khonji}, {Chau},
  {Elbassioni}, and {Zeineldin}}]{7590153}
{Karapetyan} A, {Khonji} M, {Chau} C, {Elbassioni} K, {Zeineldin} HH (2018)
  Efficient algorithm for scalable event-based demand response management in
  microgrids. IEEE Transactions on Smart Grid 9(4):2714--2725,
  \doi{10.1109/TSG.2016.2616945}

\bibitem[{Karapetyan et~al.(2021)Karapetyan, Khonji, Chau, Elbassioni,
  Zeineldin, EL-Fouly, and Al-Durra}]{9301221}
Karapetyan A, Khonji M, Chau SCK, Elbassioni K, Zeineldin H, EL-Fouly THM,
  Al-Durra A (2021) A competitive scheduling algorithm for online demand
  response in islanded microgrids. IEEE Transactions on Power Systems
  36(4):3430--3440, \doi{10.1109/TPWRS.2020.3046144}

\bibitem[{Khonji et~al.(2018)Khonji, Chau, and Elbassioni}]{khonji2016optimal}
Khonji M, Chau CK, Elbassioni K (2018) Optimal power flow with inelastic
  demands for demand response in radial distribution networks. IEEE
  Transactions on Control of Network Systems 5(1):513--524

\bibitem[{Khonji et~al.(2019)Khonji, Karapetyan, Elbassioni, and
  Chau}]{KHONJI201934}
Khonji M, Karapetyan A, Elbassioni K, Chau SCK (2019) Complex-demand scheduling
  problem with application in smart grid. Theoretical Computer Science
  761:34--50

\bibitem[{{Khonji} et~al.(2020){Khonji}, {Chau}, and {Elbassioni}}]{8892494}
{Khonji} M, {Chau} SC, {Elbassioni} K (2020) Combinatorial optimization of ac
  optimal power flow with discrete demands in radial networks. IEEE
  Transactions on Control of Network Systems 7(2):887--898,
  \doi{10.1109/TCNS.2019.2951657}

\bibitem[{Kolliopoulos and Stein(2001)}]{kolliopoulos2001approximation}
Kolliopoulos SG, Stein C (2001) Approximation algorithms for single-source
  unsplittable flow. SIAM Journal on Computing 31(3):919--946

\bibitem[{Korovesis et~al.(2004)Korovesis, Vokas, Gonos, and Topalis}]{1339347}
Korovesis PN, Vokas GA, Gonos IF, Topalis FV (2004) Influence of large-scale
  installation of energy saving lamps on the line voltage distortion of a weak
  network supplied by photovoltaic station. IEEE transactions on power delivery
  19(4):1787--1793

\bibitem[{{Lin} and {Lin}(2008)}]{4578739}
{Lin} C, {Lin} S (2008) Distributed optimal power flow with discrete control
  variables of large distributed power systems. IEEE Transactions on Power
  Systems 23(3):1383--1392, \doi{10.1109/TPWRS.2008.926695}

\bibitem[{Low(2014{\natexlab{a}})}]{low2014convex1}
Low SH (2014{\natexlab{a}}) {Convex relaxation of optimal power flow—Part I:
  Formulations and equivalence}. IEEE Transactions on Control of Network
  Systems 1(1):15--27

\bibitem[{Low(2014{\natexlab{b}})}]{low2014convex2}
Low SH (2014{\natexlab{b}}) {Convex relaxation of optimal power flow—Part II:
  Exactness}. IEEE Transactions on Control of Network Systems 1(2):177--189

\bibitem[{{Mhanna} et~al.(2016){Mhanna}, {Chapman}, and {Verbič}}]{7433473}
{Mhanna} S, {Chapman} AC, {Verbič} G (2016) A fast distributed algorithm for
  large-scale demand response aggregation. IEEE Transactions on Smart Grid
  7(4):2094--2107, \doi{10.1109/TSG.2016.2536740}

\bibitem[{Momke and Wiese(2015)}]{10.1007/978-3}
Momke T, Wiese A (2015) A (2+ epsilon)-approximation algorithm for the storage
  allocation. In: 42nd International Colloquium on Automata, Languages, and
  Programming, Springer, pp 973--984

\bibitem[{Phillips et~al.(2000)Phillips, Uma, and Wein}]{1099-1425}
Phillips CA, Uma R, Wein J (2000) Off-line admission control for general
  scheduling problems. Journal of Scheduling 3(6):365--381

\bibitem[{Raghavan and Tompson(1987)}]{Raghavan1987}
Raghavan P, Tompson CD (1987) Randomized rounding: a technique for provably
  good algorithms and algorithmic proofs. Combinatorica 7(4):365--374

\bibitem[{Srinivasan(1999)}]{Srinivasan1999}
Srinivasan A (1999) Improved approximation guarantees for packing and covering
  integer programs. SIAM Journal on Computing 29(2):648--670

\bibitem[{Zhang and Tse(2013)}]{6290429}
Zhang B, Tse D (2013) Geometry of injection regions of power networks. IEEE
  Transactions on Power Systems 28(2):788--797,
  \doi{10.1109/TPWRS.2012.2208205}

\end{thebibliography}

\appendix 

\section*{Appendix}

\section{Proof of Lemma~\ref{l4}}\label{lemma2b}
	
\begin{proof}
	For $r\in[d]$, consider the graph of the fractional profile  $\sum_{k\in\cS^q}f_k^r(e)\tilde{x}_k$ illustrated in Figure~\ref{f1}. For $p\in[P_r]$, slice the region between the horizontal axis and horizontal line at height $h^{q,p,r}$ with $\frac{1}{\epsilon}+1$ horizontal lines, with inter-distance $\epsilon h^{q,p,r}$. The intersections of the optimal profile  with these lines define a monotone function $g^{q,r}$, as pictured in Figure~\ref{f1}, with $g^{q,r}(e)\in\{l\epsilon h^{p,r}:~l\in\{0,1\ldots,1/\epsilon\},~p\in[P_r]\}$, for all $e\in\cE$. We adopt a greedy procedure, explained in Algorithm~\ref{alg:modify} below, to remove a set of demands from $\cS^q$ in each interval $\cE_p^r$ such that the remaining set of demands fractionally fits below $g^{q,r}$ (see lines~\ref{alg:s1}-\ref{alg:s7}). The algorithm proceeds by removing the ``left-most" set of demands that minimally ensures that the remaining ones in $\cS^q$ can be packed under capacity $g^{q,r}$.  This defines an intermediate fractional vector $\bar{\bar x}$ for separable $d$-USFP-R[$\cS^q,g^q$], where $g^q=(g^{q,r})_{r\in[d]}$, which can be converted to a basic feasible solution (BFS) with the same or better objective value. Lastly, the fractional components of $\bar{\bar x}$  are rounded down yielding an integral solution $\hat x$.  
	\begin{algorithm}[!htb]
		\caption{ {\sc Modify} \label{alg:modify} }
		\begin{algorithmic}[1]
			\Require  $q\in\cQ$; a restricted profile $\mbox{RP}_{\epsilon}(h;w;g)$; a set of users $\cS^q\subseteq\cI^q$; a fractional vector $(\tilde x_k)_{k\in\cS^q}\in[0,1]^{\cS^q}$
			\Ensure	A integral vector $(\hat x_k)_{k\in\cS^q}\in\{0,1\}^{\cS^q}$ satisfying conditions (i) and (ii) of lemma~\ref{l4}		
			\State $\bar x\leftarrow\tilde x$; $t\leftarrow 0$\label{alg:s2} 
			\For{$r=1$ to $d$}\label{alg:s1} 
			\For{$p=1,\ldots,P_r$}\label{alg:s1-}
			\State $i\leftarrow \underline{i}(p,r)$
			\While {$t<\epsilon h^{q,p,r}$ } \label{alg:cond} 
			\If{$\exists k\in\cS^q$ such that $\tilde x_kf_k^r(e_i)>0$}\label{alg:s3} 
			\State $\bar x_k=0$\label{alg:s4} 
			\State $t\leftarrow t+\tilde{x}_kf_k^r(e_i)$\label{alg:s5} 
			\Else{}
			$i\leftarrow i+1$\label{alg:s6} 
			\EndIf		  
			\EndWhile
			\EndFor
			\EndFor \label{alg:s7} 
			\State Convert $\bar x$ to a  BFS $\bar{\bar x}$ for $d$-USFP-R[$\cS,g^q$] with $\sum_{k\in\cS}u_k\bar{\bar x}_k\ge\sum_{k\in\cS}u_k\bar x_k$ \label{alg:s8} 
			\State$ (\hat x_k)_{k\in \cS^q} \leftarrow  \big( \lfloor \bar{\bar x}_k \rfloor \big)_{k\in \cS^q}$  \label{alg:USFP-round}
			\State \Return $\hat x$
		\end{algorithmic}
	\end{algorithm}

	We first show that condition (i) holds when $\hat x$ is replaced by $\bar x$. For $r\in[d]$, let $\cJ^{r}(e_i)$ be the set of demands $k\in\cS^q$ for which $\bar x_k$ was set to $0$ in step~\ref{alg:s4} when considering edge $e_i\in\cE$. Consider an edge $e\in\cE_p^r$ such that $\sum_{k\in\cS^q}f_k^r(e)\bar x_k>0$. Note that $0\leq \sum_{k\in\cS^q}f_k^r(e)\tilde x_k-g^{q,r}(e)\le \epsilon h^{q,p,r}$ by \raf{rel0} and the definition of $g^{q,r}$.
	By the monotonicity of $f^r_k(\cdot)$ and the condition of the while-loop in step~\ref{alg:cond} we have
	\begin{align*}
	\sum_{k\in\cS^q}f_k^r(e)\bar x_k&=\sum_{k\in\cS^q}f_k^r(e)\tilde x_k-\sum_{i:~e_i\le e}\sum_{k\in\cJ^{r}(e_i)}f_k^r(e)\tilde x_k\\
	&\le \sum_{k\in\cS^q}f_k^r(e)\tilde x_k-\sum_{i:~e_i\le e}\sum_{k\in\cJ^{r}(e_i)}f_k^r(e_i)\tilde x_k\\
	&\le \sum_{k\in\cS^q}f_k^r(e)\tilde x_k-\epsilon h^{q,p,r}\\
	&\le g^{q,r}(e).
	\end{align*}
	\begin{figure}
		\center\includegraphics[width=7.2cm]{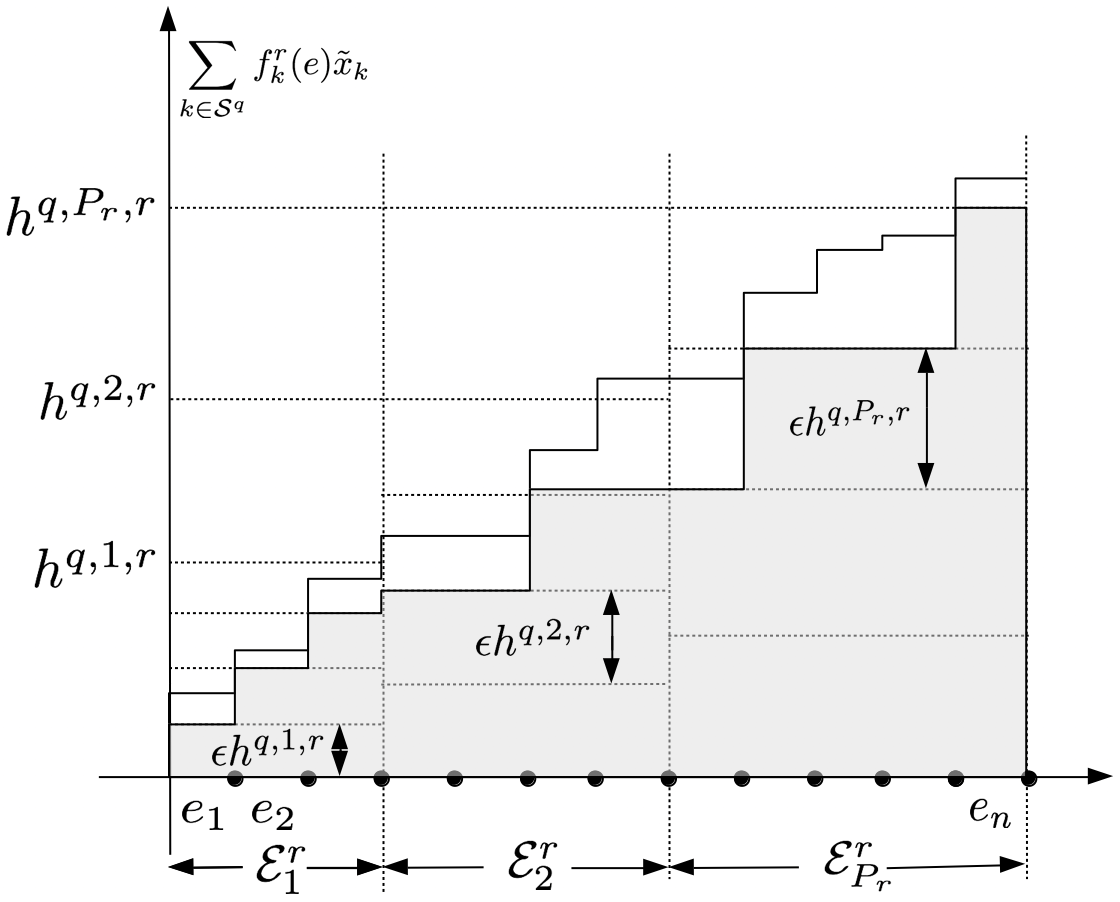}\\
		\caption{A profile and its $(h,\epsilon)$-restriction.}
		\label{f1}
	\end{figure}
	Since $\bar x$ is feasible for $d$-USFP-R[$\cS^q,g^q$], one can obtain a BFS $\bar{\bar x}$ for the same linear program with $\sum_{k\in\cS^q}u_k\bar{\bar x}_k\ge\sum_{k\in\cS}u_k\bar x_k$ as in step~\ref{alg:s8} of procedure {\sc Modify}. Then, round down the fractional components in $\bar{\bar{x}}$ to obtain an integral solution $\hat x$. Note that, for all $e\in\cE$, 
	$$\sum_{k\in\cS^q}f_k^r(e)\hat x_k\le \sum_{k\in\cS^q}f_k^r(e)\bar {\bar x}_k\le g^{q,r}(e)\le \sum_{k\in\cS^q}f_k^r(e)\tilde{x}_k,$$
	and hence (i) holds.
	
	Note that the total fractional utility of demands removed by Algorithm~\ref{alg:modify}  in steps~\ref{alg:s1}-\ref{alg:s7} is 
	\begin{align*}
	\sum_{r\in[d],~e\in\cE}\sum_{k\in\cJ^r(e)}u_k\tilde x_k&=
	\sum_{r\in[d]:~H^{P_r,r}>0}\sum_{p=1}^{P_r}\sum_{e\in\cE_p^r}\sum_{k\in\cJ^r(e)}u_k\tilde x_k\\
	&\le \sum_{r\in\cH^{q}}\sum_{p=1}^{P_r}\sum_{e\in\cE_p^r}\sum_{k\in\cJ^r(e)} \frac{1}{H^{q,p,r}} \sum_{t=1}^{T_r}a_k^{r,t}b^{r,t}(e_{\overline i(p,r)})\tilde x_k\\
	&=\sum_{r\in\cH^{q}}\sum_{p=1}^{P_r}\frac{1}{H^{q,p,r}} \sum_{e\in\cE_p^r}\sum_{k\in\cJ^r(e)} f_k^r(e_{\overline i(p,r)})\tilde x_k\\
	&\le\sum_{r\in\cH^{q}}\sum_{p=1}^{P_r}\frac{C_r}{H^{q,p,r}} \left(\sum_{e\in\cE_p^r}\sum_{k\in\cJ^r(e)} f_k^r(e)\tilde x_k\right)\\
	&\le\sum_{r\in\cH^{q}0}\sum_{p=1}^{P_r}\frac{C_r}{H^{q,p,r}}(\epsilon h^{q,p,r}+B^{q,p,r}),
	\end{align*}
	where we use the fact that $k\in\cI^q$ in the first inequality, property~\raf{property} in the second inequality, and $\underline{f}_k^{p,r}\le B^{q,p,r}$  and the condition of the while-loop in step~\ref{alg:cond} in the last inequality. (Note that we sum above over $r\in[d]$ such that in $H^{q,P_r,r}>0$ since $k\in\cJ^r(e)$ implies that $f_k^r(e)>0$, which in turn implies by~\raf{p1} that $H^{q,P_r,r}>0$.) 
	
	It follows that
	\begin{align}\label{e2}
	\sum_{k \in\cS^q}u_k\bar x_k&\ge\sum_{k \in\cS^q}u_k\tilde x_k-\sum_{r,e}\sum_{k\in\cJ^r(e)}u_k\tilde x_k \nonumber \\
	&\ge \sum_{k \in\cS^q}u_k\tilde x_k-\sum_{r\in\cH^{q}}\sum_{p=1}^{P_r}\frac{C_r}{H^{q,p,r}}(\epsilon h^{q,p,r}+B^{q,p,r}).
	\end{align}
	
	By the monotonicity of the functions $f^r_k(\cdot)$, $d$-USFP-R[$\cS,g^q$] has only $\frac{1}{\epsilon}\sum_{r=1}^dP_r$ {\it non-redundant} packing inequalities of the form \raf{ee0}. It follows that the BFS $\bar{\bar x}$ computed in step~\ref{alg:s8} has at most $\frac{1}{\epsilon}\sum_{r=1}^dP_r$ fractional components $\bar{\bar x}\in(0,1)$. Thus,
	\begin{align}\label{e2-}
	\sum_{k \in\cS^q}u_k{\hat x}_k&=\sum_{k \in\cS^q}u_k\bar{\bar x}_k-\sum_{k \in\cS^q:~\bar{\bar x}_k\in(0,1)}u_k\bar{\bar x}_k \nonumber \\
	&\ge\sum_{k \in\cS^q}u_k\bar x_k-\frac{1}{\epsilon}\sum_{r=1}^dP_r\cdot\max_ku_k\bar{\bar x}_k \nonumber\\
	&\ge\sum_{k \in\cS^q}u_k\bar{ x}_k-\frac{1}{\epsilon}\frac{\sum_{r=1}^dP_r}{\sum_{r\in\cH^{q}}P_r}\sum_{r\in\cH^{q}}\frac{P_rB^{q,P_r,r}}{H^{q,P_r,r}},
	\end{align}
	where we use in the last inequality that $\bar{\bar x}_k\le 1$ and 
	\begin{align*}
	u_k&\le \frac{\sum_{r\in\cH^{q}}P_r\sum_{t=1}^{P_r}a_k^{r,t}b^{r,t}(e_{n})/H^{q,P_r,r}}{\sum_{r\in\cH^{q}}P_r}=\frac{\sum_{r\in\cH^{q}}P_rf_k^r(e_n)/H^{q,P_r,r}}{\sum_{r\in\cH^{q}}P_r}\\ 
	& \le \frac{\sum_{r\in\cH^{q}}P_rB^{q,P_r,r}/H^{q,P_r,r}}{\sum_{r\in\cH^{q}}P_r},
	\end{align*}
	for $k\in\cS^q$. Condition (ii) follows from~\raf{e2} and~\raf{e2-}. \hfill$\blacksquare$
\end{proof}

\section{Proof of Lemma~\ref{lem:exact}}\label{lem44}

\begin{algorithm}[!htb]
	\caption{ {\sc Forward-Backward-Sweep} \label{alg:fbsweep} }
	\begin{algorithmic}[1]
		\Require  A feasible solution  $F'=(s_0', x', v',\ell',S')$ to {\sc cOPF$'[x']$} such that $\ell_{h,t}' > \frac{|S_{h,t}'|^2}{v_h'}$ for some $(h,t)\in \cE$
		\Ensure	A feasible solution $\tilde F=(\tilde s_0, \tilde x, \tilde v, \tilde \ell, \tilde S)$ to {\sc cOPF$'[x']$} such that $\tilde x = x'$ and $\sum_{e\in\cE}\tilde\ell_e<\sum_{e\in\cE}\ell_e'$
		\State $\tilde x \leftarrow x'$; $\tilde v_0 \leftarrow v_0$ \label{alg:fbs0}
		\State Number nodes $\cV=\{0,1,\ldots, m\}$ in a breadth-first search order
		\For{$j=m,m-1,\ldots,1$}  /* Forward sweep */
		\State Let $i$ be s.t. $(i,j) \in \cE$
		\State  $\tilde \ell_{i,j} \leftarrow \frac{|S_{i,j}'|^2}{v_i'}$   \label{alg:fbs1}
		\State $\tilde S_{i,j} \leftarrow \sum_{k \in \cN_j} s_k \tilde x_k  + \sum_{t:(j,t)\in \cE} \tilde S_{j,t} + z_{i,j} \tilde \ell_{i,j}$\label{alg:fbs2}
		\EndFor 
		\State $\tilde s_0 \leftarrow - \tilde S_{0,1}$ \label{alg:fbs3}
		\For{$j = 1,2,\ldots, m$} /* Backward sweep */
		\State Let $i$ be s.t. $(i,j) \in \cE$ 
		\State $\tilde v_j \leftarrow \tilde v_i + |z_{i,j}|^2 \tilde \ell_{i,j} - 2 \re(z_{i,j}^\ast  \tilde S_{i,j}) $  \label{alg:fbs4}
		\EndFor
		\State \Return $\tilde F$
	\end{algorithmic}
\end{algorithm}

\begin{proof}
	The analysis follows the same lines as in \citep{gan2015exact, low2014convex2,huang2017sufficient} and is sketched here for completeness. 
	Let $F''=(s_0'', x', v'',\ell'',S'')$ be an optimal solution of {\sc cOPF}$[x']$, which can be found (to within any desired accuracy) in polynomial time, by solving a convex program. Consider the following problem. 
	\begin{align}
	\textsc{(cOPF$'[x']$)} \quad &\min_{\substack{s_0, x, v,\ell, S \;\;}}\sum_{e\in\cE}\ell_e,  &\notag \\
	\text{s.t.} \ &\raf{p1:con2}-\raf{p1:con7}, \raf{p1:con9}, \raf{p2:con1}\notag &\\
	& x=x'&\\
	& f_{\textsc{OPF}}(s_0, x)\ge f_{\textsc{OPF}}(s_0'',x'). \label{p3:con1} &
	\end{align}
	Clearly,  \textsc{cOPF$'[x']$}  is feasible as $F''$ satisfies all its constraints. Hence, it has an optimal solution $F'=(s_0', x', v',\ell',S')$, which we claim satisfies the statement of the lemma.  
	Suppose, for the sake of contradiction, that there exists an edge $(h,t)$ such that $\ell_{h,t}' > \tfrac{|S_{h,t}'|^2}{v'_h}$.
	In the sequel, we  construct a feasible solution $\tilde F=(\tilde s_0, x', \tilde v, \tilde \ell, \tilde S)$ for \textsc{cOPF$'[x']$} such that  $\sum_{e\in\cE}\tilde\ell_e<\sum_{e\in\cE}\ell_e'$, leading to a contradiction.
	
	Apply the forward-backward sweep algorithm, illustrated in Alg.~\ref{alg:fbsweep}, on the solution $F'$ to obtain a feasible solution $\tilde F$. 
	
	We show the feasibility of the solution $\tilde F$.
	By Steps~\ref{alg:fbs2}, \ref{alg:fbs3} and \ref{alg:fbs4} of Alg.~\ref{alg:fbsweep}, all equality constraints of \textsc{cOPF$'[x']$}  are satisfied. 
	By Step~\ref{alg:fbs1} and the feasibility of $F'$, we also have 
	\begin{align}\label{e-l-}
	\tilde \ell_e \le \ell_e'\le \overline \ell_e\text{ for all  $e \in \cE$.}
	\end{align}
	Next, by rewriting  $\tilde S_{i,j}$, recursively substituting from the leaves, we get
	\begin{align}\label{et1}
	\tilde S_{i,j}=\sum_{k \in \cN_j} s_k \tilde x_k + \sum_{e \in \cE_j\cup\{(i,j)\}}z_e \tilde \ell_{e}.
	\end{align}

	Write $\Delta \ell_e \triangleq \tilde \ell_e - \ell_e'\le 0$, $\Delta S_e \triangleq \tilde S_e - S_e'$, and $\Delta |S_e|^2 \triangleq  |\tilde S_e|^2 - |S_e'|^2$, for  $e\in \cE$. Let $\widehat S_{j} \triangleq  \sum_{k \in \cN_j} s_k x_k'$,  $\tilde L_{i,j} \triangleq \sum_{e \in \cE_j\cup\{(i,j)\}}z_e \tilde \ell_e$, and $ L'_{i,j} \triangleq \sum_{e \in \cE_j\cup\{(i,j)\}}z_e  \ell_e'$. Note by \raf{et1} that $\tilde S_{i,j} = \widehat S_{j} + \tilde L_{i,j}$ and,  similarly, $ S_{i,j} = \widehat S_{j} +  L_{i,j}'$. It follows that, for all $(i,j)\in\cE$,
	\begin{align}\label{eq:lem1.0}
	\Delta S_{i,j} &=  \tilde L_{i,j}-L_{i,j}'=\sum_{e \in \cE_j\cup\{(i,j)\}}z_{e} \Delta \ell_{e} \le 0,
	\end{align}
	where the inequality follows by assumption {\sf A1}. 
	In particular, for $(i,j)=(0,1)$, we obtain 
	\begin{align}\label{s-e-}
	\tilde s_0^{\rm R} &= -\tilde S_{0,1}^{\rm R}  \ge - S_{0,1}'^{\rm R} = s_0'^{\rm R},
	\end{align}
	implying by {\sf A0} that $f_0(\tilde s_0^{\rm R})  \ge f_0(s_0'^{\rm R})$ and hence \raf{p3:con1} is satisfied.

	Furthermore,
	\begin{align}
	\Delta |S_{i,j}|^2 &= |\tilde S_{i,j}|^2 - |S_{i,j}'|^2\\
	&= (\tilde S_{i,j}^{\rm R})^2 -  (S_{i,j}'^{\rm R})^2 +(\tilde S_{i,j}^{\rm I})^2 - (S_{i,j}'^{\rm I})^2 \\
	&=\Delta S_{i,j}^{\rm R} (\tilde S_{i,j}^{\rm R} +  S_{i,j}'^{\rm R}) + \Delta S_{i,j}^{\rm I} (\tilde S_{i,j}^{\rm I} +  S_{i,j}'^{\rm I})\\
	&=\sum_{e \in \cE_j\cup\{(i,j)\}} z_e^{\rm R} \Delta \ell_e \big( 2 \widehat S^{\rm R}_{j} + \tilde L_{i,j}^{\rm R} + L_{i,j}'^{\rm R} \big) \nonumber \\
	&\quad+ \sum_{e \in \cE_j\cup\{(i,j)\}} z_e^{\rm I} \Delta \ell_e \big( 2 \widehat S^{\rm I}_{j} + \tilde L_{i,j}^{\rm I} + L_{i,j}'^{\rm I} \big) \\
	&= \sum_{e \in \cE_j\cup\{(i,j)\}}  2\Delta \ell_e \re(z_e^*\widehat S_j)+ \sum_{e \in \cE_j\cup\{(i,j)\}}  \Delta \ell_e \re(z_e^*\tilde L_{i,j})\nonumber\\
	&\quad+\sum_{e \in \cE_j\cup\{(i,j)\}}  \Delta \ell_e\re(z_e^*L'_{i,j})\le 0, \label{eq:lem1.1}
	\end{align}
	where Eqn.~\raf{eq:lem1.1} follows by {\sf A1}, {\sf A3} (or {\sf A4$'$}) and $\Delta \ell_e \le 0$. Therefore, by the feasibility of $S_e$, 
	\begin{equation}\label{eq:mS}
	|\tilde{S}_{e}| \le |S_{e}'| \le \overline S_{e} \quad \text{for all } e \in \cE .
	\end{equation}
	Note that, by~{\sf A1}, the inequalities in ~\raf{eq:mS} also imply that the reverse power constraint in~\raf{p1:con6} is satisfied for $\tilde S$.
	
	Rewrite Cons.~\raf{p1:con4}  by recursively substituting $\tilde v_j$, for $j$ moving away from the root, and then substituting for $\tilde S_{h,t}$ using \raf{et1}:
	\begin{align}
	\tilde v_j &= v_0 -2\sum_{ (h,t) \in \cP_{j} } \re(z_{h,t}^\ast  \tilde S_{h,t}) +  \sum_{( h,t) \in \cP_{j} } |z_{h,t}|^2 \tilde \ell_{h,t}\notag\\
	&=v_0 -2 \sum_{ (h,t) \in \cP_{j} } \re\Big(z^*_{h,t}\big(\sum_{k\in \cN_t}  s_k \tilde x_k + \sum_{e \in \cE_t\cup\{(h,t)\}} z_{e} \tilde \ell_{e} \big)\Big) +   \sum_{ (h,t) \in \cP_{j} } |z_{h,t}|^2 \tilde \ell_{h,t},\notag\\
	&=v_0 -2 \sum_{k \in \cN} \re\Big( \sum_{(h,t)\in \cP_k \cap \cP_j} z^*_{h,t} s_k\Big) \tilde x_k   - 2 \sum_{(h,t)\in \cP_j} \re \Big( z_{h,t}^* \sum_{e \in \cE_t} z_{e} \tilde \ell_{e}\Big)  \notag \\ 
	&~~-    2\sum_{ (h,t) \in\cP_{j} }|z_{h,t}|^2 \tilde \ell_{h,t} +\sum_{ (h,t) \in\cP_{j} }|z_{h,t}|^2 \tilde \ell_{h,t},  \label{eq:vLL}
	\end{align}
	where the last statement follows from exchanging the summation operators, and $z^*_{e} z_{e} = |z_{e}|^2$. Thus,
	\begin{align}
	\tilde v_j&=v_0 -2 \sum_{k \in \cN} \re\Big( \sum_{(h,t)\in \cP_k \cap \cP_j} z^*_{h,t} s_k\Big) \tilde x_k \nonumber\\
	&\quad - \Big(2 \sum_{(h,t)\in \cP_j} \re \big( z_{h,t}^* \sum_{e \in \cE_t} z_{e} \tilde \ell_{e}\big) +    \sum_{ (h,t) \in\cP_{j} }|z_{h,t}|^2 \tilde \ell_{h,t}\Big) \nonumber \\
	&\le v_0 <\overline v_j, \label{eq:vL0}
	\end{align}
	where the first inequality follows by {\sf A1} and {\sf A3}, and the last inequality follows by {\sf A2}. 
	Since  $\tilde \ell_e \le \ell'_e$ and $\tilde x = x'$, we get by {\sf A1} and the feasibility of $F'$,
	\begin{align}
	\tilde v_j	&\ge v_0 -2 \sum_{k \in \cN} \re\Big( \sum_{(h,t)\in \cP_k \cap \cP_j} z^*_{h,t} s_k\Big)  x_k' \nonumber \\
	&\quad - \Big(2 \sum_{(h,t)\in \cP_j} \re \big( z_{h,t} \sum_{e \in \cE_t} z_{e}  \ell_{e}\big) +    \sum_{ (h,t) \in\cP_{j} }|z_{h,t}|^2  \ell_{h,t}\Big) \nonumber \\
	&=v_j'\ge \underline v_j. \label{eq:vL2}
	\end{align}
	By Ineqs.~\raf{eq:mS} and \raf{eq:vL2}, $\tilde \ell_{i,j} = \frac{| S_{i,j} '|^2}{v_i'} \ge \frac{| \tilde S_{i,j} |^2}{\tilde v_i}$, hence, $\tilde F$ is feasible.
	
	Finally, by the first inequality in \raf{e-l-} and the fact that $\ell_{h,t} '> \tfrac{|S_{h,t}|^2}{v_h}=\tilde\ell_{h,t} $, we have $\sum_{e\in\cE}\tilde \ell_{e}<\sum_{e\in\cE} \ell_{e}'$, contradicting the optimality of $F'$ for \textsc{cOPF$'[x']$}. \hfill$\blacksquare$

\end{proof}

\section{Proof of Lemma~\ref{feas}}\label{lprf}

\begin{proof}
	The argument is similar to that in Lemma~\ref{lem:exact}. 
	We apply a slightly modified version of Alg.~\ref{alg:fbsweep} on the solution $F'$ to obtain a feasible solution $\tilde F$.  Replace steps~\ref{alg:fbs0} and~\ref{alg:fbs1} in Alg.~\ref{alg:fbsweep}, respectively, by: 
	\begin{align}\label{e-msteps}
	\text{1: $\tilde x \leftarrow \bar x$; $\tilde v_0 \leftarrow v_0$, and 5: $\tilde \ell_{i,j} \leftarrow \ell'_{i,j}$.}
	\end{align}
	By Steps~\ref{alg:fbs2}, \ref{alg:fbs3} and \ref{alg:fbs4} of the (modified) algorithm, all equality constraints of \textsc{(rcOPF$[\bar x]$)}  are satisfied. 
	By (modified) Step~\ref{alg:fbs1} and the feasibility of $F'$, we also have 
	\begin{align}\label{e-l}
	\tilde \ell_e =\ell_e'\le \overline \ell_e\text{ for all  $e \in \cE$.}
	\end{align}
	Write $\Delta S_e \triangleq \tilde S_e - S_e'$, and $\Delta |S_e|^2 \triangleq  |\tilde S_e|^2 - |S_e'|^2$, for  $e\in \cE$. Let $S_{j}' \triangleq  \sum_{k \in \cN_j} s_k x_k'$, $\tilde S_{j}\triangleq  \sum_{k \in \cN_j} s_k \tilde x_k$, and   $\tilde L_{i,j} \triangleq \sum_{e \in \cE_j\cup\{(i,j)\}}z_e \tilde \ell_e$. Note by \raf{et1} that $\tilde S_{i,j} = \tilde S_{j} + \tilde L_{i,j}$ and, $ S_{i,j}' =  S_{j}' +  \tilde L_{i,j}$. It follows that, for all $(i,j)\in\cE$,
	\begin{align}\label{eq:lem1.0.1}
	\Delta S_{i,j} &=  \tilde S_{j}-S_{j}'=\sum_{k \in \cN_j} s_k \bar x_k-\sum_{k \in \cN_j} s_k x_k' \le 0,
	\end{align}
	where the inequality follows from \raf {feas:con2} and \raf{feas:con3}. 
	In particular, for $(i,j)=(0,1)$, we obtain 
	\begin{align}\label{s-e}
	\tilde s_0^{\rm R} &= -\tilde S_{0,1}^{\rm R}  \ge - S_{0,1}'^{\rm R} = s_0'^{\rm R},
	\end{align}
	implying by {\sf A0$'$} that $f_0(\tilde s_0^{\rm R} \cos \phi + \tilde s_0^{\rm I} \sin \phi)  \ge f_0(s_0'^{\rm R} \cos \phi + s_0'^{\rm I} \sin \phi))$ and hence $f_{\textsc{OPF}}(\tilde s_0,\tilde x)\ge(1-\eps)f_{\textsc{OPF}}(s_0',x')$ follows from \raf{feas:con0} and \raf{feas:con4}.
	
	Furthermore,
	\begin{align*}
	\Delta |S_{i,j}|^2 &= |\tilde S_{i,j}|^2 - |S_{i,j}'|^2\\
	&= (\tilde S_{i,j}^{\rm R})^2 -  (S_{i,j}'^{\rm R})^2 +(\tilde S_{i,j}^{\rm I})^2 - (S_{i,j}'^{\rm I})^2 \\
	&=\Delta S_{i,j}^{\rm R} (\tilde S_{i,j}^{\rm R} +  S_{i,j}'^{\rm R}) + \Delta S_{i,j}^{\rm I} (\tilde S_{i,j}^{\rm I} +  S_{i,j}'^{\rm I})\\
	&=\Delta S_{i,j}^{\rm R}(\tilde S_j^{\rm R}+S_j'^{\rm R}+2\tilde L_{i,j}^{\rm R})+ \Delta S_{i,j}^{\rm I}(\tilde S_j^{\rm I}+S_j'^{\rm I}+2\tilde L_{i,j}^{\rm I})\le 0, 
	\end{align*}
	where the last inequality follows by {\sf A1}, {\sf A4$'$} and \raf{eq:lem1.0.1}. Therefore, 
	\begin{equation}\label{eq:bS}
	|\tilde{S}_{i,j}| \le |S'_{i,j}| \le \overline S_{i,j}.
	\end{equation}
	
	Next, we show $\underline v_j \le \tilde v_j \le \overline v_j$. 	As in~\raf{eq:vLL}, rewrite Cons.~\raf{p1:con4}  by recursively substituting $v_j'$, for $j$ moving away from the root,  and then substituting for $\tilde S_{h,t}$ using \raf{et1}:
	\begin{align}\label{eq:ew1}
	v_j'&=v_0 -2 \sum_{k \in \cN} \re\Big( \sum_{(h,t)\in \cP_k \cap \cP_j} z^*_{h,t} s_k\Big)  x_k' \nonumber \\
	& \quad- \Big(2 \sum_{(h,t)\in \cP_j} \re \big( z_{h,t}^* \sum_{e \in \cE_t} z_{e}  \ell'_{e}\big) +    \sum_{ (h,t) \in\cP_{j} }|z_{h,t}|^2  \ell'_{h,t}\Big)   
	\end{align}
	A similar equation can be derived for $\tilde v_j$, where $x'$ and $\ell'$ in \raf{eq:ew1} are replaced by $\tilde x$ and $\tilde\ell$, respectively. By assumptions {\sf A2} and {\sf A3}, we have
	\begin{align*}
	\tilde v_j&=v_0 -2 \sum_{k \in \cN} \re\Big( \sum_{(h,t)\in \cP_k \cap \cP_j} z^*_{h,t} s_k\Big) \tilde x_k \nonumber\\
	&\quad - \Big(2 \sum_{(h,t)\in \cP_j} \re \big( z_{h,t}^* \sum_{e \in \cE_t} z_{e} \tilde \ell_{e}\big) +    \sum_{ (h,t) \in\cP_{j} }|z_{h,t}|^2 \tilde \ell_{h,t}\Big) \nonumber\\
	&\le v_0 <\overline v_j.
	\end{align*}
	Moreover, since  $\tilde \ell_e = \ell'_e$ and $\tilde x = \bar x$ satisfies~\raf{feas:con1}, we get by {\sf A1} and the feasibility of $F'$,
	\begin{align}\label{eq:bV}
	\tilde v_j
	&\ge v_0 -2 \sum_{k \in \cN} \re\Big( \sum_{(h,t)\in \cP_k \cap \cP_j} z^*_{h,t} s_k\Big)  x_k' \nonumber \\
	&\quad - \Big(2 \sum_{(h,t)\in \cP_j} \re \big( z_{h,t} \sum_{e \in \cE_t} z_{e}  \ell_{e}'\big) +    \sum_{ (h,t) \in\cP_{j} }|z_{h,t}|^2  \ell_{h,t}'\Big) \nonumber \\
	&=v_j'\ge \underline v_j.
	\end{align}
	
	Finally, by inequalities~\raf{eq:bS} and \raf{eq:bV}, $\tilde \ell_{i,j} = \ell_{i,j}' = \frac{|S'_{i,j}|^2}{v'_i} \ge \frac{|\tilde S_{i,j}|^2}{\tilde v_i}$, hence $\tilde \ell_{i,j}$ satisfies Cons.~\raf{p2:con1}. \hfill$\blacksquare$
\end{proof}

\end{document}